\documentclass[onecolumn]{IEEEtran}
\usepackage{amsmath,amssymb,bm,cite,color,amsthm,mathrsfs}
\usepackage[thicklines]{cancel}
\usepackage{stfloats}
\usepackage{multirow}
\usepackage{enumerate}
\usepackage[linesnumbered,ruled]{algorithm2e}
\usepackage{hyperref}
\usepackage{cleveref}

\newtheorem{theorem}{Theorem}[section]
\newtheorem{lemma}{Lemma}[section]

\newtheorem{proposition}{Proposition}[section]
\newtheorem{corollary}{Corollary}[section]
\newtheorem{definition}{Definition}[section]

\newtheorem{remark}{Remark}[section]

\newcommand\nc\newcommand
\nc{\cA}{\mathcal{A}}\nc{\cB}{\mathcal{B}}\nc{\cC}{\mathcal{C}}\nc{\cD}{\mathcal{D}}
\nc{\cE}{\mathcal{E}}\nc{\cF}{\mathcal{F}}\nc{\cG}{\mathcal{G}}\nc{\cH}{\mathcal{H}}
\nc{\cI}{\mathcal{I}}\nc{\cJ}{\mathcal{J}}\nc{\cK}{\mathcal{K}}\nc{\cL}{\mathcal{L}}
\nc{\cM}{\mathcal{M}}\nc{\cN}{\mathcal{N}}\nc{\cO}{\mathcal{O}}\nc{\cP}{\mathcal{P}}
\nc{\cQ}{\mathcal{Q}}\nc{\cR}{\mathcal{R}}\nc{\cS}{\mathcal{S}}\nc{\cT}{\mathcal{T}}
\nc{\cU}{\mathcal{U}}\nc{\cV}{\mathcal{V}}\nc{\cW}{\mathcal{W}}\nc{\cX}{\mathcal{X}}
\nc{\cY}{\mathcal{Y}}\nc{\cZ}{\mathcal{Z}}

\nc{\bba}{\mathbf{a}}\nc{\bbb}{\mathbf{b}}\nc{\bbc}{\mathbf{c}}\nc{\bbd}{\mathbf{d}}
\nc{\bbe}{\mathbf{e}}\nc{\bbf}{\mathbf{f}}\nc{\bbg}{\mathbf{g}}\nc{\bbh}{\mathbf{h}}
\nc{\bbi}{\mathbf{i}}\nc{\bbj}{\mathbf{j}}\nc{\bbk}{\mathbf{k}}\nc{\bbl}{\mathbf{l}}
\nc{\bbm}{\mathbf{m}}\nc{\bbn}{\mathbf{n}}\nc{\bbo}{\mathbf{o}}\nc{\bbp}{\mathbf{p}}
\nc{\bbq}{\mathbf{q}}\nc{\bbr}{\mathbf{r}}\nc{\bbs}{\mathbf{s}}\nc{\bbt}{\mathbf{t}}
\nc{\bbu}{\mathbf{u}}\nc{\bbv}{\mathbf{v}}\nc{\bbw}{\mathbf{w}}\nc{\bfx}{\mathbf{x}}
\nc{\bby}{\mathbf{y}}\nc{\bbz}{\mathbf{z}}

\nc{\bbA}{\mathbf{A}}\nc{\bbB}{\mathbf{B}}\nc{\bbC}{\mathbf{C}}\nc{\bbD}{\mathbf{D}}
\nc{\bbE}{\mathbf{E}}\nc{\bbF}{\mathbf{F}}\nc{\bbG}{\mathbf{G}}\nc{\bbH}{\mathbf{H}}
\nc{\bbI}{\mathbf{I}}\nc{\bbJ}{\mathbf{J}}\nc{\bbK}{\mathbf{K}}\nc{\bbL}{\mathbf{L}}
\nc{\bbM}{\mathbf{M}}\nc{\bbN}{\mathbf{N}}\nc{\bbO}{\mathbf{O}}\nc{\bbP}{\mathbf{P}}
\nc{\bbQ}{\mathbf{Q}}\nc{\bbR}{\mathbf{R}}\nc{\bbS}{\mathbf{S}}\nc{\bbT}{\mathbf{T}}
\nc{\bbU}{\mathbf{U}}\nc{\bbV}{\mathbf{V}}\nc{\bbW}{\mathbf{W}}\nc{\bfX}{\mathbf{X}}
\nc{\bbY}{\mathbf{Y}}\nc{\bbZ}{\mathbf{Z}}

\nc{\sA}{\mathsf{A}}\nc{\sB}{\mathsf{B}}\nc{\sC}{\mathsf{C}}\nc{\sD}{\mathsf{D}}
\nc{\sE}{\mathsf{E}}\nc{\sF}{\mathsf{F}}\nc{\sG}{\mathsf{G}}\nc{\sH}{\mathsf{H}}
\nc{\sI}{\mathsf{I}}\nc{\sJ}{\mathsf{J}}\nc{\sK}{\mathsf{K}}\nc{\sL}{\mathsf{L}}
\nc{\sM}{\mathsf{M}}\nc{\sN}{\mathsf{N}}\nc{\sO}{\mathsf{O}}\nc{\sP}{\mathsf{P}}
\nc{\sQ}{\mathsf{Q}}\nc{\sR}{\mathsf{R}}\nc{\sS}{\mathsf{S}}\nc{\sT}{\mathsf{T}}
\nc{\sU}{\mathsf{U}}\nc{\sV}{\mathsf{V}}\nc{\sW}{\mathsf{W}}\nc{\sX}{\mathsf{X}}
\nc{\sY}{\mathsf{Y}}\nc{\sZ}{\mathsf{Z}}

\nc{\bma}{\bm{a}}\nc{\bmb}{\bm{b}}\nc{\bmc}{\bm{c}}\nc{\bmd}{\bm{d}}
\nc{\bme}{\bm{e}}\nc{\bmf}{\bm{f}}\nc{\bmg}{\bm{g}}
\nc{\bmh}{\bm{h}}\nc{\bmi}{\bm{i}}\nc{\bmj}{\bm{j}}
\nc{\bmk}{\bm{k}}\nc{\bml}{\bm{l}}\nc{\bmm}{\bm{m}}
\nc{\bmn}{\bm{n}}\nc{\bmo}{\bm{o}}\nc{\bmp}{\bm{p}}
\nc{\bmq}{\bm{q}}\nc{\bmr}{\bm{r}}\nc{\bms}{\bm{s}}\nc{\bmt}{\bm{t}}\nc{\bmu}{\bm{u}}\nc{\bmv}{\bm{v}}\nc{\bmw}{\bm{w}}\nc{\bmx}{\bm{x}}\nc{\bmy}{\bm{y}}\nc{\bmz}{\bm{z}}

\nc{\bmA}{\bm{A}}\nc{\bmB}{\bm{B}}\nc{\bmC}{\bm{C}}\nc{\bmD}{\bm{D}}
\nc{\bmE}{\bm{E}}\nc{\bmF}{\bm{F}}\nc{\bmG}{\bm{G}}
\nc{\bmH}{\bm{H}}\nc{\bmI}{\bm{I}}\nc{\bmJ}{\bm{J}}
\nc{\bmK}{\bm{K}}\nc{\bmL}{\bm{L}}\nc{\bmM}{\bm{M}}
\nc{\bmN}{\bm{N}}\nc{\bmO}{\bm{O}}\nc{\bmP}{\bm{P}}
\nc{\bmQ}{\bm{Q}}\nc{\bmR}{\bm{R}}\nc{\bmS}{\bm{S}}\nc{\bmT}{\bm{T}}\nc{\bmU}{\bm{U}}\nc{\bmV}{\bm{V}}\nc{\bmW}{\bm{W}}\nc{\bmX}{\bm{X}}\nc{\bmY}{\bm{Y}}\nc{\bmZ}{\bm{Z}}

\newcommand{\mathset}[1]{\left\{#1\right\}}
\newcommand{\multiset}[1]{\left\{\left\{#1\right\}\right\}}
\newcommand{\abs}[1]{\left|#1\right|}
\newcommand{\ceilenv}[1]{\left\lceil #1 \right\rceil}
\newcommand{\floorenv}[1]{\left\lfloor #1 \right\rfloor}
\newcommand{\parenv}[1]{\left( #1 \right)}
\newcommand{\sparenv}[1]{\left[ #1 \right]}

\nc{\set}[1]{\llbracket #1 \rrbracket}

\newcommand{\vt}{\mathsf{VT}}

\newcommand{\Rep}{\textup{Rep}}

\newcommand{\cfcite}[2]{\textit{c.f.}~\cite[#2]{#1}}
\newcommand{\cfcitep}[1]{\textit{c.f.}~\cite{#1}}

\title{Deletion-Correcting Codes for the $\ell$-Symbol Read Channel}
\author{Zuo~Ye and Gennian Ge%
\thanks{This research is supported by the National Key Research and Development Program of China under Grant 2025YFC3409900, the National
Natural Science Foundation of China under Grant 12231014 and Grant 12501466, Beijing Scholars Program, and Xiaomi Young Scholars Program.

Z. Ye is with the Institute of Mathematics and Interdisciplinary Sciences, Xidian University,
Xi'an 710126, China. Email: yezuo@xidian.edu.cn.

G. Ge is with the School of Mathematical Sciences, Capital Normal University, Beijing 100048, China, Email: gnge@zju.edu.cn.
}
}

\begin{document}
\maketitle

\begin{abstract}
    This paper studies deletion-correcting codes for the $\ell$-symbol read channel, whose noiseless output is the vector of all consecutive $\ell$-mers of a transmitted sequence. This model is motivated by overlapping-read mechanisms arising in nanopore sequencing, racetrack memories with consecutive read heads, and related sequence-labeling problems. We consider an adversarial setting in which a fixed number of $\ell$-mers are deleted from the read vector. Our first contribution is a structural characterization of the effect of such deletions: after a minimum number of $\ell$-mers are inserted to restore consistency, the resulting sequence is obtained from the transmitted sequence by deleting symbols from certain periodic substrings; when $t\le \ell-2$, these deletions correspond to complete minimum periods. Based on this characterization, we introduce check patterns and construct $\ell$-read deletion-correcting codes via power-sum syndromes. For every $\ell\ge2$, we obtain single-deletion correcting codes with redundancy $\log\lfloor (n+2\ell)/(\ell-1)\rfloor$. For $2\le t\le \ell/2$, we construct $q$-ary $\ell$-read $t$-deletion correcting codes with redundancy $t\log n+O(1)$, and for $\ell=2t-1$ with $t\ge3$, we construct codes with redundancy $(2t-1)\log n+O(1)$. We also study the sporadic parameter pairs $(\ell,t)\in\{(2,2),(3,2),(3,3)\}$ and obtain improved constructions, including binary $\ell$-read $2$-deletion correcting codes with redundancy $2\log n+O(1)$ for $\ell=2,3$, a non-binary $3$-read $2$-deletion correcting code with redundancy $3\log n+O(1)$, a binary $3$-read $3$-deletion correcting code with redundancy $5\log n+O(1)$, and a non-binary $3$-read $3$-deletion correcting code with redundancy $7\log n+O(1)$.
\end{abstract}


\section{Introduction}
We consider an abstract readout model in which a sequence is observed through overlapping windows of length $\ell$. More precisely, for a sequence $\bmx=x_1x_2\cdots x_n$, the noiseless channel output is the vector
$$
\cR_{\ell}(\bmx)
\triangleq\parenv{\bmx_{[1,\ell]},\bmx_{[2,\ell+1]},\ldots,\bmx_{[n-\ell+1,n]}},
$$
where $\bmx_{[i,j]}\triangleq x_ix_{i+1}\cdots x_{j}$ for all $i\le j$. We study an adversarial deletion model for this channel: the decoder receives a subvector of $\cR_{\ell}(\bmx)$, obtained by deleting a prescribed number, say $t$, of $\ell$-symbol reads, and aims to recover the original sequence. A code capable of correcting such errors is called an $\ell$-read $t$-deletion correcting code. Throughout the paper, $\ell$ and $t$ are treated as constants.

This model is motivated by several readout mechanisms in sequencing and storage systems. In nanopore sequencing, the measured signal depends on a local block of consecutive nucleotides rather than on a single nucleotide. In symbol-pair and $b$-symbol read channels, adjacent or consecutive groups of symbols are read together. In DNA labeling, using a sufficiently rich set of labels allows one to identify local substrings of a DNA strand. In racetrack memory, multiple read heads can observe local groups of stored symbols, while shift errors may delete or repeat read positions. These different settings lead to the same basic mathematical object: an overlapping read vector followed by synchronization errors. We next discuss these motivations and the most relevant prior work.

\subsection{Motivation and Related Work}
\subsubsection{\textbf{Nanopore Sequencing}}
DNA has emerged as a promising medium for archival data storage owing to its exceptionally high storage density and long-term stability \cite{Omer2024TMBMSC,Milenkovic2024TCOMM}. In a DNA-based data storage system, data is first encoded into sequences over the DNA nucleotide alphabet $\mathset{A,C,G,T}$. The corresponding DNA strands are then synthesized and stored. To access the original data, the stored DNA strands are sequenced and converted back into nucleotide sequences. Among available sequencing technologies, nanopore sequencing is particularly attractive because of its advantages including relatively low cost, high portability, and the ability to produce much longer reads \cite{Omer2024TMBMSC,Deamer2016}.

In a nutshell, nanopore sequencing reads a single-stranded DNA molecule by driving it through a nanopore and recording ionic-current measurements as the nucleotides translocate through the pore. These measurements are subsequently processed by a basecalling algorithm to infer the underlying DNA sequence. Importantly, each measurement is typically influenced not by a single nucleotide, but by several, say $\ell$, consecutive nucleotides within the sensing region of the pore. This local-context dependence is commonly abstracted by a function $f:\mathset{A,C,G,T}^{\ell}\rightarrow \cY$, where $\cY$ is the output alphabet. Under this abstraction, the noiseless readout corresponding to an input DNA sequence $\bmx=x_1x_2\cdots x_n$ is the vector $\cR_{f}(\bmx)\triangleq\parenv{f(\bmx_{[1,\ell]}),f(\bmx_{[2,\ell+1]}),\ldots,f(\bmx_{[n-\ell+1,n]})}$ \cite{Hulett2021ISIT,Omer2024TMBMSC}. In practice, the observed readout corresponding to $\cR_{f}(\bmx)$ may be affected by various errors, including deletions, tandem duplications, and substitutions \cite{Maowei2018IT}. The goal is to reconstruct $\bmx$ from such a corrupted version of $\cR_{f}(\bmx)$.

The preceding channel model naturally extends to arbitrary $q$-ary alphabets. Coding problems for this channel have been studied under different choices of the map $f$ and different error models. For example, when $f$ maps an $\ell$-mer to to its composition, several works have considered the correction of substitution \cite{Banerjee2023ISIT,Banerjee2024IT,Yubo2025ITnanopore,Banerjee2025ISIT}, deletion \cite{Banerjee2024ISIT}, or duplication errors \cite{WenjunYu2025IT} in $\cR_{f}(\bmx)$. Chee \emph{et al} \cite{CheeYeowMeng2024ISIT} studied constrained coding schemes for correcting deletions or duplications in $\cR_{f}(\bmx)$, without specifying the map $f$.

The works most directly related to the present paper are the two papers by Xie and Chen \cite{Zitan2025ITW,Zitan202601}, where it is assumed that one has access to a sequence of $\ell$-mers inferred by basecallers. Under this assumption, the noiseless channel output corresponding to $\bmx$ is $\cR_{\ell}(\bmx)$. Notice that $\cR_{\ell}(\bmx)$ is equivalent to $\cR_f(\bmx)$ when $f$ is injective. They further assume that $\bmx$ is prefixed with a known $\ell$-mer $\bmu$ and suffixed with a known $\ell$-mer $\bmv$ before entering the nanopore. Thus, the exact transmitted sequence is $\bmu\bmx\bmv$. This assumption is practically reasonable in nanopore sequencing, where adapter sequences are commonly attached to DNA strands before sequencing \cite[Section II]{Maowei2018IT}. The goal of their works is to design large codes enabling the recovery of $\bmx$ even when $t$ components of $\cR_{\ell}\parenv{\bmu\bmx\bmv}$ are deleted. For $t=2$, Xie and Chen in \cite{Zitan2025ITW} constructed a $q$-ary code with redundancy $3\log(n)+O(1)$ for constant $\ell\ge4$. Later, in \cite{Zitan202601}, they constructed a $q$-ary code with redundancy $2t\log(n)+o(\log(n))$ for general $t$, under the condition that the received corrupted version of $\cR_{\ell}\parenv{\bmu\bmx\bmv}$ is \emph{consistent} and that $t\le\min\mathset{\ell-2,(\ell+1)/2}$. While this paper was in preparation, Xie and Chen improved the results of \cite{Zitan202601} and presented a construction with redundancy $t\log(n)+O(\log\log(n))$ under the condition $t\le\min\mathset{(\ell-1)/2,(\ell+2)/3}$ \cite{Zitan202606}. They also showed that the minimum redundancy of $\ell$-read $t$-deletion correcting codes is at least $t\log(n)+\Omega(1)$.

\vspace{5pt}
\subsubsection{\textbf{$\ell$-Symbol Read Channel}}
The symbol-pair read channel was introduced by Cassuto and Blaum to model storage systems in which adjacent symbols are read together rather than individually \cite{CassutoBlaum2011IT}.
It was later generalized to the $\ell$-symbol read channel, where each read consists of $\ell$ consecutive symbols \cite{YaakobiBruckSiegel2016IT}. Given a sequence $\bmx=x_1x_2\cdots x_{n}$, the noiseless channel output of the $\ell$-symbol read channel is the cyclic $\ell$-symbol read vector
$\pi_{\ell}(\bmx)\triangleq\parenv{\bmx_{[1,\ell]},\bmx_{[2,\ell+1]},\ldots,x_{n}x_{1}\cdots x_{\ell-1}}$.

Most prior work on symbol-pair and $\ell$-symbol read channels has focused on substitution errors in the read vector.  Synchronization errors for the symbol-pair read channel were studied in \cite{YeowVan2020ISIT,CheeYeowMeng2021IT}. In particular, \cite{YeowVan2020ISIT} constructed binary codes correcting one or two deletions in the $2$-symbol read vector, with redundancies $\log n+O(1)$ and $4\log n+o(\log n)$, respectively.

Let $x_{-\ell+1},\ldots,x_0,x_{n+1}$ be $\ell+1$ known symbols. By definition, it holds that $$
\pi_{\ell}\parenv{x_{-\ell+1}\cdots x_0\bmx x_{n+1}}=\cR_{\ell}\parenv{x_{-\ell+1}\cdots x_0\bmx x_{n+1}x_{-\ell+1}\cdots x_{-1}}.
$$
Therefore, the $\ell$-symbol read channel studied in the previous literature is essentially equivalent to the channel considered in this paper. For this reason, we refer to the channel studied here as the $\ell$-symbol read channel. We note that this channel is also called the acyclic $\ell$-symbol read channel in \cite[Section VII-B]{CheeYeowMeng2021IT}.

\vspace{5pt}
\subsubsection{\textbf{DNA Labeling}}
DNA labeling is a tool for visualizing, detecting, and studying DNA at the molecular level. Let $\bmx \in\{A,C,G,T\}^n$ and let $\cS=\{\alpha_1,\ldots,\alpha_M\}$ be a set of labels, where each $\alpha_i \in \{A,C,G,T\}^{\ell}$. The labeling  output can be modeled as a labeling sequence
$\cL_{\cS}(\bmx) \in \{0,1,\ldots,M\}^n$, where the $i$-th coordinate of $\cL_{\cS}(\bmx)$ equals $j\in\{1,\ldots,M\}$ if $\bmx_{[i,i+\ell-1]}=\alpha_j$, and equals $0$ if no label in $\cS$ appears at position $i$. Thus, each output symbol is determined by a length-$\ell$ substring of the underlying DNA sequence.

Clearly, when $\cS=\{A,C,G,T\}^{\ell}$, exactly one label matches each substring $\bmx_{[i,i+\ell-1]}$. Hence the labeling output uniquely identifies the sequence of consecutive $\ell$-mers. In this case, the labeling sequence is essentially
\[
\bigl(\bmx_{[1,\ell]}, \bmx_{[2,\ell+1]}, \ldots, \bmx_{[n-\ell+1,n]}\bigr),
\]
which is exactly the $\ell$-symbol read vector $\cR_{\ell}(\bmx)$ defined previously.

The capacity of DNA labeling and the minimum number of labels needed to reconstruct almost all DNA sequences have been studied in \cite{Hanania2025IT,HofmeisterChristoph2026IT}. Error-correcting codes for labeled DNA sequences were investigated in \cite{HananiaYaakobi2025ISIT}, where substitution, deletion, and insertion errors in the labeling sequence were considered. In particular, the work \cite{HananiaYaakobi2025ISIT} generalized the binary $2$-read $1$-deletion correcting code in \cite{YeowVan2020ISIT} to general alphabets.

\vspace{5pt}
\subsubsection{\textbf{Racetrack Memory}}
Racetrack memory is an emerging non-volatile storage technology in which binary data is stored in magnetic domains arranged along a nanoscopic wire, also called a racetrack or track. Each magnetic domain stores one bit through its magnetization direction. During the read process, the read heads remain fixed, while the domains are shifted along the track so that different domains can be aligned with the heads. In typical architectures, multiple heads may be placed along the racetrack strip, often uniformly, to reduce read latency \cite{ZhangChao2015ISCA,CheeKiahVardy2018IT}.

This architecture is attractive because of its high storage density, but the shift operation may be imprecise. If the track is shifted too far, some domains may be skipped; if it is not shifted far enough, the same domain may be read more than once. These position errors are commonly modeled as deletions and sticky insertions, respectively. Several works \cite{CheeKiahVardy2018IT,CheeYeowMeng2018ITW,SimaBruck2023IT} have studied codes correcting such errors by exploiting the multi-head structure of racetrack memories, especially under the condition that the distance between adjacent heads is $\Theta(\log(n))$, where $n$ is the code length.

Let
$\bmx=x_1x_2\cdots x_n\in\{0,1\}^n$
be the stored data. Suppose that there are $\ell$ consecutive read heads. Then, in the absence of shift errors, the simultaneous read at time $i$ is $x_i x_{i+1}\cdots x_{i+\ell-1}$.
Hence the noiseless readout is $\parenv{\bmx_{[1,\ell]},\bmx_{[2,\ell+1]},\ldots,\bmx_{[n-\ell+1,n]}}=\cR_{\ell}(\bmx)$\footnote{The works \cite{CheeKiahVardy2018IT,CheeYeowMeng2018ITW,SimaBruck2023IT} assume that each head can read the entire sequence $\bmx$, whereas here we assume that the $i$-th head can read only $\bmx_{[i,n-\ell+i]}$. This boundary difference is immaterial for our purposes, since known prefixes and suffixes can be appended to the stored sequence.}. Thus, the $\ell$-symbol read channel studied in this paper can be viewed as the adjacent-head case of the multi-head racetrack readout model. Under this interpretation, an over-shift that skips one read position deletes one component of $\cR_{\ell}(\bmx)$, while an under-shift repeats one component and gives rise to a sticky insertion.

Consequently, the results of this paper give deletion-correcting codes for racetrack memories with adjacent read heads. The codes in \cite{CheeKiahVardy2018IT,CheeYeowMeng2018ITW,SimaBruck2023IT} achieve much lower redundancy and apply to a broader range of parameters $\ell$ and $t$. However, they require a much larger distance between adjacent read heads, which increases the area overhead of the device.

\subsection{Our Contribution}
Motivated by the applications mentioned above, this paper continues the study of deletion correction for the $\ell$-symbol read channel. We now summarize the main ideas and contributions of the paper.

First, we establish a structural connection between deletions in the $\ell$-mer read vector and deletions in the transmitted sequence itself. Given a received subvector $\cR^\prime$ of $\cR_{\ell}\parenv{\bmu\bmx\bmv}$, we insert the minimum possible number of $\ell$-mers into $\cR'$ to obtain a consistent vector $\cR^{\prime\prime}$, and then derive a sequence $\bmy$ from it. We show that the reconstructed sequence $\bmy$ is obtained from $\bmu\bmx\bmv$ by a collection of burst deletions. Moreover, when $t\le \ell-2$, these deletions occur inside periodic substrings and correspond to deleting some complete minimum periods. This characterization is the main structural ingredient of our constructions.

Second, we introduce the notion of an $(\ell,t,s)$-check pattern, which captures the periodic substrings that can be shortened. We prove that, under suitable parameter conditions, the relevant check patterns in $\bmu\bmx\bmv$ and $\bmy$ have the same number, the same order, and the same minimum periods. Consequently, the decoding problem reduces to identifying which check patterns suffered deletions and how many complete periods were deleted from each pattern.

Third, using power-sum syndromes, we construct several general families of deletion-correcting codes. For every $\ell\ge2$, we obtain $q$-ary $\ell$-read $1$-deletion correcting codes with redundancy $\log\floorenv{(n+2\ell)/(\ell-1)}$.
For $2\le t\le \ell/2$, we construct $q$-ary $\ell$-read $t$-deletion correcting codes with redundancy $t\log n+O(1)$. We also treat the boundary-type case $\ell=2t-1$ for $t\ge3$ and obtain codes with redundancy $(2t-1)\log n+O(1)$.

Fourth, we study several small parameter pairs that are not covered by the general constructions. For $(\ell,t)=(2,2)$, Chee and Vu constructed a binary $2$-read $2$-deletion correcting code with redundancy $4\log n+o(\log n)$ \cite[Theorem~9]{YeowVan2020ISIT}. We reduce this redundancy to $2\log n+O(1)$ in the binary case, and also give a non-binary construction with redundancy $4\log n+O(1)$. For $(\ell,t)=(3,2)$, we obtain a binary construction with redundancy $2\log n+O(1)$ and a non-binary construction with redundancy $3\log n+O(1)$. Finally, for $(\ell,t)=(3,3)$, we construct binary and non-binary codes with redundancies $5\log n+O(1)$ and $7\log n+O(1)$, respectively. These sporadic constructions show that the sufficient condition $\ell\ge 2t$ in the general construction is not necessary for achieving redundancy $t\log n+O(1)$, and that the joint condition $\ell=2t-1$ with $t\ge3$ is not necessary for achieving redundancy $(2t-1)\log n+O(1)$.

Finally, based on these constructions, we provide alternative constructions with efficient encoders at the cost of additional $O\parenv{\log\log(n)}$ redundant bits.

\subsubsection{\textbf{Comparison}}
While this paper was in preparation, Xie and Chen updated their work \cite{Zitan202601} and gave an improved construction of $\ell$-read $t$-deletion correcting codes with redundancy $t\log(n)+O\parenv{\log\log(n)}$ under the condition $t\le\min\mathset{\frac{\ell-1}{2},\frac{\ell+2}{3}}$ \cite{Zitan202606}. In comparison, we attain the same redundancy under the mild condition $\ell/2\ge t\ge 1$. We list the main differences between the two constructions below.
\begin{itemize}
    \item Both constructions are based on the fact that deletions in $\cR_{\ell}(\bmu\bmx\bmv)$ only shorten certain periodic substrings of $\bmu\bmx\bmv$. Their construction works with \emph{check patterns}, which are periodic substrings with minimum period at most $t$ and length at least $\ell-t$. Inspired by the observation that the decoder can determine the  the number $s$ of $\ell$-mers inserted into $\cR^\prime$, we work with finer patterns, called $(\ell,t,s)$-\emph{check patterns}, which are periodic substrings with minimum period at most $t-s$ and length at least $\ell-1-s$. This enables us to obtain constructions for more flexible parameters $\ell$ and $t$.
    \item Their construction uses $B_t$ sets to compute a single syndrome, while our construction uses $t$ power-sum syndromes.
    \item The encoding and decoding algorithms are different.
\end{itemize}

For the reader's convenience, we summarize the state-of-the-art results on $\ell$-read $t$-deletion correcting codes in \Cref{tab_lreaddelcode}.

\begin{table*}[!htbp]
    \centering
    \caption{Summary of results on $\ell$-read $t$-deletion correcting codes}
    \label{tab_lreaddelcode}
    \begin{tabular}{c|c|c|c}
    \hline
    \hline
      $(\ell,t)$& Alphabet Size $q$&Redundancy& Reference \\
       \hline
        \multirow{2}{*}{$(2,1)$}&$q=2$&$\log n+O(1)$ &\cite[Theorem 6]{YeowVan2020ISIT}\\
        \cline{2-4}
        &$q\ge2$&$\log n+O(1)$&\cite[Theorem 3]{HananiaYaakobi2025ISIT}\\
        \hline
        $(2,2)$ &$q=2$&$4\log(n)+o(\log(n))$ &\cite[Theorem 9]{YeowVan2020ISIT}\\
        \hline
        $t\le\min\mathset{\frac{\ell-1}{2},\frac{\ell+2}{3}}$&$q\ge2$&$t\log(n)+O(\log\log(n))$&\cite[Theorem 15]{Zitan202606}\\
        \hline
        $(\ell\ge4,2)$& $q\ge2$ & $3\log(n)+O(1)$ &\textbf{\cite[Theorem 5]{Zitan2025ITW}}\\
        \hline
        $(\ell\ge2,1)$ &$q\ge2$&$\log\parenv{\floorenv{(n+2\ell)/(\ell-1)}}$ & \textbf{\Cref{thm_1del}}\\
        \hline
        $\ell\ge2t$, $t\ge2$&$q\ge2$& $t\log(n)+O(1)$ &\textbf{\Cref{thm_multidel1}}\\
        \hline
        $(2t-1,t\ge3)$& $q\ge2$& $(2t-1)\log(n)+O(1)$ &\textbf{\Cref{thm_multidel2}}\\
        \hline
        \multirow{2}{*}{$(2,2)$} & $q\ge2$&$2\log(n)+O(1)$&\textbf{\Cref{thm_binaryl2t2}}\\
        \cline{2-4}
        &$q>2$&$4\log(n)+O(1)$&\textbf{\Cref{thm_qaryl2t2}}\\
        \hline
        \multirow{2}{*}{$(3,2)$}&$q=2$&$2\log(n)+O(1)$&\textbf{\Cref{cor_bianryl3t2}}\\
        \cline{2-4}
        &$q>2$&$3\log(n)+O(1)$&\textbf{\Cref{thm_qaryl3t2}}\\
        \hline
        \multirow{2}{*}{$(3,3)$}&$q=2$&$5\log(n)+O(1)$&\textbf{\Cref{thm_binaryl3t3}}\\
        \cline{2-4}
        &$q>2$&$7\log(n)+O(1)$&\textbf{\Cref{thm_qaryl3t3}}\\
        \hline
        \hline
    \end{tabular}
\end{table*}

\subsection{Organization}
The rest of the paper is organized as follows. In \Cref{sec_preliminary}, we introduce the $\ell$-symbol read channel, define $\ell$-read deletion-correcting codes, and prove structural results relating $\ell$-mer deletions to deletions from periodic substrings of the transmitted sequence. In \Cref{sec_genconstruction}, we develop the general check-pattern framework and present the single-deletion and multiple-deletion constructions. In \Cref{sec_sporadic}, we study the sporadic parameter pairs $(2,2)$, $(3,2)$, and $(3,3)$, deriving improved constructions for both binary and non-binary alphabets. For completeness, we present codes with efficient encoders and decoders in \Cref{sec_explicitcode}. Finally, \Cref{sec_conclusion} concludes the paper and discusses several directions for future work.

\section{Preliminary}\label{sec_preliminary}
This section introduces the notation, terminology, and auxiliary results that will be frequently used throughout the paper.

For an integer $q\ge 2$, let $\Sigma_q$ denote the $q$-ary alphabet $\mathset{0,1,\ldots,q-1}$. We use ordinary lowercase letters, such as $a,b,c,x,y,z$, for symbols in $\Sigma_q$, and bold lowercase letters, such as $\bma,\bmb,\bmc,\bmx,\bmy,\bmz$, for sequences over $\Sigma_q$. The concatenation of two sequences $\bmx$ and $\bmy$ is denoted by $\bmx\bmy$. Unless otherwise stated, the $i$-th symbol of a sequence $\bmx$ is denoted by $x_i$. The \emph{length} of a sequence $\bmx$, denoted by $\abs{\bmx}$, is the number of symbols in $\bmx$. The set of all sequences of length $n$ over $\Sigma_q$ is denoted by $\Sigma_q^n$. We also use $\abs{S}$ to denote the cardinality of a finite set $S$.

For integers $m$ and $n$ with $m\le n$, let $[m,n]$ denote the set $\mathset{m,m+1,\ldots,n}$. When $n\ge1$, the set $[1,n]$ is also denoted by $[n]$. For a sequence $\bmx\in\Sigma_q^n$ and a set $I=\mathset{i_1,\ldots,i_s}\subseteq[n]$, define $\bmx_{I}=x_{i_1}x_{i_2}\cdots x_{i_s}$, and call it a \emph{subsequence} of $\bmx$. In particular, when $I=[a,b]$ for some $1\le a\le b\le n$, we call $\bmx_{I}$ a \emph{substring} of $\bmx$. A $t$-\emph{burst-deletion} in $\bmx$ is the deletion of a length-$t$ substring from $\bmx$.

\subsection{Channel Model}
The channel studied in this paper is called the \emph{$\ell$-symbol read channel}. In the absence of errors, this channel takes a sequence $\bmx\in\Sigma_q^n$ as input and outputs the vector of $\ell$-\emph{mers}: $\parenv{\bmx_{[1,\ell]},\bmx_{[2,\ell+1]},\ldots,\bmx_{[n-\ell+1,n]}}$. Here, an $\ell$-mer is a sequence in $\Sigma_q^{\ell}$.

For the sake of theoretical analysis (see the proof of \Cref{lem_burstdel}), we assume throughout this paper that  two \emph{known} $\ell$-mers, denoted by $\bmu$ and $\bmv$, are appended to both
ends of $\bmx$ before transmission. Thus, the exact transmitted sequence is $\tilde{\bmx}\triangleq\bmu\bmx\bmv$. For clarity, the coordinates of $\tilde{\bmx}$ are indexed by the set $\mathset{-\ell+1,\ldots,0,1,\ldots,n,n+1,\ldots,n+\ell}$ instead of $[n+2\ell]$. Equivalently, $\tilde{\bmx}_{[-\ell+1,0]}=\bmu$, $\tilde{\bmx}_{[n]}=\bmx$, and $\tilde{\bmx}_{[n+1,n+\ell]}=\bmv$.

Under this convention, the channel output associated with $\bmx$ is $\cR_{\ell}\parenv{\tilde{\bmx}}\triangleq\parenv{\bms_0,\bms_1,\ldots,\bms_{n+\ell}}$, where $\bms_i=\tilde{\bmx}_{[i-\ell+1,i]}$ for all $0\le i\le n+\ell$. For every $0\le i<n+\ell$, the length-$(\ell-1)$ suffix of $\bms_i$ equals the length-$(\ell-1)$ prefix of $\bms_{i+1}$. This property motivates the following definition.

\begin{definition}
    Let $\bmy=y_1\cdots y_{\ell}$ and $\bmz=z_1\cdots z_{\ell}$ be two $\ell$-mers. We say that they are \emph{consistent} if $y_2\cdots y_{\ell}=z_1\cdots z_{\ell-1}$. Otherwise, they are called \emph{inconsistent}.
    Let $\cR=\parenv{\bms_1,\ldots,\bms_m}$ be a vector of $\ell$-mers. We say that $\cR$ is consistent if $\bms_i$ and $\bms_{i+1}$ are consistent for every $1\le i<m$. Otherwise, it is called inconsistent.
\end{definition}

Given a consistent vector of $\ell$-mers $\parenv{\bms_1,\ldots,\bms_m}$, where $\bms_i=s_{i,1}\cdots s_{i,\ell}$ for $i\in[m]$, let $\bmx=s_{1,1}\cdots s_{1,\ell}s_{2,\ell}\cdots s_{m,\ell}$. Then we have $\cR_{\ell}(\bmx)=\parenv{\bms_1,\ldots,\bms_m}$. We therefore denote the reconstructed sequence $\bmx$ by $\cR_{\ell}^{-1}\parenv{\bms_1,\ldots,\bms_m}$. Clearly, we have $\cR_{\ell}^{-1}\parenv{\cR_{\ell}(\tilde{\bmx})}=\tilde{\bmx}$. Hence the original sequence $\bmx$ can be recovered directly whenever the error-free vector $\cR_{\ell}(\tilde{\bmx})$ is received.

In practice, however, the decoder may receive a corrupted version, say $\cR^\prime$, of $\cR_{\ell}(\tilde{\bmx})$. In this paper, we focus on deletion errors, namely the case where $\cR^\prime$ is a subvector of $\cR_{\ell}(\tilde{\bmx})$. Equivalently, $\cR^\prime$ is obtained from $\cR_{\ell}(\tilde{\bmx})$ by deleting some $\ell$-mers. In this situation, $\cR^\prime$ is not necessarily consistent. Recall that $\tilde{\bmx}_{[-\ell+1,0]}=\bmu$ and $\tilde{\bmx}_{[n+1,n+\ell]}=\bmv$. Thus, without loss of generality, we always assume that the first and last $\ell$-mers in $\cR_{\ell}(\tilde{\bmx})$ are not deleted.

\begin{definition}
    Let $\cC$ be a nonempty subset of $\Sigma_q^n$. We call $\cC$ an $\ell$-read $t$-deletion correcting code, if for any $\bmx\in\cC$ and any length-$(n+\ell+1-t)$ subvector $\cR^\prime$ of $\cR_{\ell}\parenv{\tilde{\bmx}}$, we can efficiently and uniquely decode $\bmx$ from $\cR^\prime$.
\end{definition}

Clearly, a $1$-mer $t$-deletion correcting code is the ordinary \emph{$t$-deletion correcting code}, which has been extensively investigated in recent years.

The redundancy of a code $\cC\subseteq\Sigma_q^n$, denoted by $\rho\parenv{\cC}$, is defined to be $\log\parenv{q^n/\abs{\cC}}$, where $\log(\cdot)$ denotes the base-two logarithm. Throughout this paper, the integers $q$, $\ell$, and $t$ are assumed to be constants.

\subsection{Auxiliary Results}
We next establish several auxiliary results that will be used repeatedly in the rest of this paper.

\begin{definition}
Let $\bmz$ be a sequence of length at least $T+1$. We say that $\bmz$ is a \emph{$T$-periodic sequence}, if it has \emph{period} $T$, i.e.,  $z_i=z_{i+T}$ for all $1\le i\le\abs{\bmz}-T$. If $T$ is the smallest positive integer with this property, we say that $T$ is the \emph{minimum period} of $\bmz$.

A substring of $\bmz$ is called a \emph{maximal $T$-periodic substring} if it is $T$-periodic and is not contained in a longer $T$-periodic substring of $\bmz$.
\end{definition}

Let $\cR^\prime$ be a subvector of $\cR_{\ell}(\tilde{\bmx})$ of length $n+\ell+1-t$. In \Cref{cor_arbitrarydel} below, we show that one can still derive from $\cR^\prime$ a sequence $\bmy$, which is obtained from $\tilde{\bmx}$ by at most $t$ deletions. Moreover, if $t\le\ell-2$, these deletions only shorten some periodic substrings.

We begin with the case in which $\cR^\prime$ is obtained from $\cR_{\ell}(\tilde{\bmx})$ by a single burst of deletions. The following lemma generalizes \cite[Lemma 1]{Zitan2025ITW} and \cite[Lemma 4]{Zitan202601}.
\begin{lemma}\label{lem_burstdel}
    Let $\bmx\in\Sigma_q^n$, $\ell\ge2$, and $t\ge1$. Suppose that $\cR^\prime$ is obtained from $\cR_{\ell}(\tilde{\bmx})$ by deleting $\bms_{i+1},\ldots,\bms_{i+t}$, where $0\le i<n+\ell-t$.
    \begin{enumerate}[$(i)$]
        \item There is a smallest integer $s$ with $0\le s\le\min\mathset{\ell-1,t}$ such that inserting $s$ $\ell$-mers $\bmz_1,\ldots,\bmz_{s}$ into $\cR^\prime$ yields a consistent vector $\cR^{\prime\prime}$ of $\ell$-mers. The value of $s$ can be determined directly from $\cR^\prime$. The sequence $\cR_{\ell}^{-1}\parenv{\cR^{\prime\prime}}$ is obtained from $\tilde{\bmx}$ by a $(t-s)$-burst-deletion occurring in the substring $\tilde{\bmx}_{[i-\ell+2+s,i+t-s]}$; more precisely, by deleting the substring $\tilde{\bmx}_{[i+1,i+t-s]}$.
        \item If $s<\min\mathset{\ell-1,t}$, then this $(t-s)$-burst-deletion occurs in a maximal periodic substring of length at least $\ell-1+t-2s$ whose minimum period is at most $t-s$.
        \item If $s<t\le\ell-2$, then the minimum period of the maximal periodic substring in (ii) divides $t-s$.
    \end{enumerate}
\end{lemma}
\begin{IEEEproof}
    By assumption, the only possible inconsistent pair of $\ell$-mers in $\cR^\prime$ is $(\bms_i,\bms_{i+t+1})$. Recall that $\bms_i=\tilde{\bmx}_{[i-\ell+1,i]}$ and $\bms_{i+t+1}=\tilde{\bmx}_{[i+t-\ell+2,i+t+1]}$.

    If $\cR^\prime$ is consistent, then $\bms_i$ and $\bms_{i+t+1}$ are consistent. Then $s=0$ necessarily. In this case, $\cR^{\prime\prime}=\cR^\prime$ and $\cR_{\ell}^{-1}\parenv{\cR^{\prime\prime}}=\tilde{x}_{-\ell+1}\cdots \tilde{x}_0\cdots\tilde{x}_i\tilde{x}_{i+t+1}\cdots \tilde{x}_{n}\tilde{x}_{n+1}\cdots \tilde{x}_{n+\ell}$, which is obtained from $\tilde{\bmx}$ by deleting $\tilde{\bmx}_{[i+1,i+t]}$.

    If $\cR^\prime$ is inconsistent, then $\parenv{\bms_i,\bms_{i+t+1}}$ is the only pair of adjacent $\ell$-mers in $\cR^\prime$ that are inconsistent. Thus, this pair can be identified immediately.  If there is an integer $r$ with $0\le r\le\ell-2$ such that
    \begin{equation}\label{eq_burstdel}
        \tilde{x}_{i-\ell+2+r}\cdots \tilde{x}_{i}=\tilde{x}_{i-\ell+t+2}\cdots \tilde{x}_{i+t-r},
    \end{equation}
    let $s$ be the smallest such $r$. For $j=1,\ldots,s$, define $\bmz_j=\tilde{x}_{i-\ell+1+j}\cdots \tilde{x}_i\tilde{x}_{i+t-s+1}\cdots \tilde{x}_{i+t-s+j}$. Let
    $$
    \cR^{\prime\prime}=\parenv{\bms_0,\ldots,\bms_i,\bmz_1,\ldots,\bmz_s,\bms_{i+t+1},\ldots,\bms_{n+\ell}}.
    $$
    Then $\cR^{\prime\prime}$ is consistent. Clearly, $\cR_{\ell}^{-1}\parenv{\cR^{\prime\prime}}=\tilde{x}_{-\ell+1}\cdots \tilde{x}_0\cdots \tilde{x}_i\tilde{x}_{i+t-s+1}\cdots \tilde{x}_{n}\tilde{x}_{n+1}\cdots \tilde{x}_{n+\ell}$, which is obtained from $\tilde{\bmx}$ by deleting $\tilde{\bmx}_{[i+1,i+t-s]}$.

    Note that if $t\le\ell-2$, we have $i-\ell+2\le i-\ell+t+2\le i$, implying that such an integer $s$ must exist and satisfy $s\le t$.

    Now suppose that no integer $r$ with $0\le r\le\ell-2$ satisfies \eqref{eq_burstdel}. Then we must have $t\ge\ell-1$. Set $s=\ell-1$. For $1\le j\le\ell-1$, define $\bmz_j=\tilde{x}_{i-\ell+1+j}\cdots \tilde{x}_i\tilde{x}_{i-\ell+t+2}\cdots \tilde{x}_{i-\ell+t+1+j}$. Let
    $$
    \cR^{\prime\prime}=\parenv{\bms_0,\ldots,\bms_i,\bmz_1,\ldots,\bmz_{\ell-1},\bms_{i+t+1},\ldots,\bms_{n+\ell}}.
    $$
    Then $\cR^{\prime\prime}$ is consistent. Clearly, $\cR_{\ell}^{-1}\parenv{\cR^{\prime\prime}}=\tilde{x}_{-\ell+1}\cdots \tilde{x}_0\cdots \tilde{x}_i\tilde{x}_{i-\ell+t+2}\cdots \tilde{x}_{n}\tilde{x}_{n+1}\cdots \tilde{x}_{n+\ell}$, which is obtained from $\tilde{\bmx}$ by deleting the substring $\tilde{\bmx}_{[i+1,i-\ell+t+1]}$, which is of length $t-\ell+1=t-s$.

    The preceding argument also shows that the burst-deletion occurs in the substring $\tilde{\bmx}_{[i-\ell+2+s,i+t-s]}$, which has length $\ell-1+t-2s$. This proves (i).

    When $s<\min\mathset{\ell-1,t}$, \Cref{eq_burstdel} implies that $t-s$ is a period of the substring $\tilde{\bmx}_{[i-\ell+2+s,i+t-s]}$. Therefore, the minimum period of $\tilde{\bmx}_{[i-\ell+2+s,i+t-s]}$ is at most $t-s$. This proves (ii).

    If $t\le\ell-2$, we have $\ell-1-s+(t-s)\ge 2(t-s)$. Since $\tilde{\bmx}_{[i-\ell+2+s,i+t-s]}$ is $(t-s)$-periodic, its minimum period must divide $t-s$. This proves (iii).
\end{IEEEproof}

 This lemma motivates the following definition.
\begin{definition}
    Suppose that $\bmp$ is a periodic sequence with minimum period $T$. Let $\bmp^\prime$ be obtained from $\bmp$ by several bursts of deletions. If the length of each burst is divisible by $T$, we say that $\bmp^\prime$ is obtained from $\bmp$ by deleting several complete minimum periods.
\end{definition}

Therefore, \Cref{lem_burstdel} shows that when $s<t\le\ell-2$, the sequence $\cR_{\ell}^{-1}\parenv{\cR^{\prime\prime}}$ is obtained from $\tilde{\bmx}$ by deleting several complete minimum periods of a periodic substring of length at least $\ell-1+t-2s$.

The preceding lemma extends naturally to the case of multiple bursts of deletions.
\begin{theorem}\label{thm_multiburst}
    Let $\bmx\in\Sigma_q^n$. Suppose that $\cR^\prime$ is obtained from $\cR_{\ell}(\tilde{\bmx})$ by deleting $\bms_{i_j+1},\ldots,\bms_{i_j+t_j}$ for $j\in[k]$, where $t_1,\ldots,t_k\ge1$ and $i_{j+1}-i_j\ge t_{j}+2$ for all $1\le j<k$. Equivalently, $\cR^\prime$ is obtained from $\cR_{\ell}(\tilde{\bmx})$ by $k$ non-overlapping bursts of deletions. Let $t=\sum_{j=1}^kt_j$. Denote $\cR^\prime=\parenv{\bms_0^\prime,\bms_1^\prime,\ldots,\bms_{n+\ell-t}}$. Let $m_j=\sum_{r=1}^{j-1}t_j$ for all $j\in[k]$.
    \begin{enumerate}[$(i)$]
        \item There are $k$ integers $s_1,\ldots,s_k$ with $0\le s_j\le\min\mathset{\ell-1,t_j}$ for each $j\in[k]$, such that inserting $s_j$ $\ell$-mers between $\bms_{i_j-m_j}^\prime$ and $\bms_{i_j-m_j+1}^\prime$ yields a consistent vector $\cR^{\prime\prime}$ of $\ell$-mers. These integers can be determined directly from $\cR^\prime$. The sequence $\cR_{\ell}^{-1}\parenv{\cR^{\prime\prime}}$ is obtained from $\tilde{\bmx}$ by deleting the $k$ substrings $\tilde{\bmx}_{[i_1+1,i_1+t_1-s_1]},\ldots,\tilde{\bmx}_{[i_k+1,i_k+t_k-s_k]}$.
        \item If $s_j<\min\mathset{\ell-1,t_j}$, then $\tilde{\bmx}_{[i_j+1,i_j+t_j-s_j]}$ is contained in a maximal periodic substring of length at least $\ell-1+t_j-2s_j$ whose minimum period is at most $t_j-s_j$.
        \item If $s_j<t_j\le\ell-2$, the minimum period of the maximal periodic substring in (ii) divides $t_j-s_j$.
    \end{enumerate}
\end{theorem}
\begin{IEEEproof}
    It follows from the definition of $\cR^\prime$ that $\bms_{i_j-m_j}^\prime=\bms_{i_j}^\prime$ and $\bms_{i_j-m_j+1}^\prime=\bms_{i_j+t_j+1}$ for all $j\in[k]$. The $k$ pairs of $\ell$-mers $\parenv{\bms_{i_1-m_1}^\prime,\bms_{i_1-m_1+1}^\prime},\ldots,\parenv{\bms_{i_k-m_k}^\prime,\bms_{i_k-m_k+1}^\prime}$ are the only possibly inconsistent pairs of $\ell$-mers in $\cR^\prime$. Applying \Cref{lem_burstdel} to each of these pairs proves the theorem.
\end{IEEEproof}

\begin{remark}
    Notice that there may be different $j$'s such that $s_{j}<\min\mathset{\ell-1,t_{j}}$, and different $\tilde{\bmx}_{[i_j+1,i_j+t_j-s_j]}$ could be contained in the same maximal periodic substring whose minimum period is at most $t_{j}-s_j$. In such a case, we interpret these $(t_j-s_j)$-burst-deletions as being sequentially applied to the same maximal periodic substring. Thus, claim (ii) says that before the deletion of each $\tilde{\bmx}_{[i_j+1,i_j+t_j-s_j]}$, the maximal periodic substring has length at least $\ell-1+t_j-2s_j$.
\end{remark}

We now turn to arbitrary deletions in $\cR_{\ell}(\tilde{\bmx})$.
\begin{corollary}\label{cor_arbitrarydel}
    Let $\bmx\in\Sigma_q^n$ and $t\ge1$. Suppose that $\cR^\prime$ is obtained from $\cR_{\ell}(\tilde{\bmx})$ by deleting $t$ arbitrary $\ell$-mers.
    \begin{enumerate}[$(i)$]
        \item We can insert a minimum number $s$ of $\ell$-mers into $\cR^\prime$ to get a consistent vector $\cR^{\prime\prime}$ of $\ell$-mers, where $0\le s\le t$. The value of $s$ can be determined directly.
        \item The sequence $\cR_{\ell}^{-1}\parenv{\cR^{\prime\prime}}$ is obtained from $\tilde{\bmx}$ by some bursts of deletions. The total length of these bursts is $t-s$.
        \item If $s<\min\mathset{\ell-1,t}$, each of the burst-deletions in (ii) occurs in a maximal periodic substring whose minimum period is at most $t-s$.
        \item If $s<t\le\ell-2$, the minimum period of each of these maximal periodic substrings in (iii) divides the number of deleted symbols in that substring. In other words, each of these substrings undergoes a deletion of several complete minimum periods.
        \item Let $\bmp$ be one of the maximal periodic substrings in (iii), and let $\bmp^\prime$ be the sequence corresponding to $\bmp$ after deletions. Then $\abs{\bmp^\prime}\ge\ell-1-s$.
        \item Let $\bmp$ and $\bmp^\prime$ be as in (iv). The minimum period of $\bmp$ is also a period of $\bmp^\prime$.
    \end{enumerate}
\end{corollary}
\begin{IEEEproof}
    Notice that deletions in $\cR_{\ell}(\tilde{\bmx})$ can be partitioned into several non-overlapping bursts of deletions. Then (i)-(v) follow immediately from \Cref{thm_multiburst}. Regarding (vi), suppose that the minimum period of $\bmp$ is $T$. According to (iv), $\bmp^\prime$ is obtained from $\bmp$ by deleting several complete minimum periods. Consequently, $\bmp^\prime$ is a substring of $\bmp$. By (v), we have $\abs{\bmp^\prime}\ge\ell-1-s>t-s\ge T$. Therefore, $T$ is also a period of $\bmp^\prime$. This completes the proof.
\end{IEEEproof}

\section{General Constructions}\label{sec_genconstruction}
In this section, we construct $\ell$-read $t$-deletion correcting codes with redundancy $t\log(n)+O(1)$ for every pair $(\ell,t)$ satisfying $1\le t\le\ell/2$, and codes with redundancy $(2t-1)\log(n)+O(1)$ for every pair $(\ell,t)$ satisfying $\ell=2t-1$ and $t\ge3$.

Suppose that $\cR^\prime$ is obtained from $\cR_{\ell}(\tilde{\bmx})$ by deleting $t$ arbitrary $\ell$-mers, and that a consistent vector $\cR^{\prime\prime}$ of $\ell$-mers is obtained from $\cR^\prime$ by inserting a minimum number $s$ of $\ell$-mers. If $s<t\le\ell-2$, then \Cref{cor_arbitrarydel} implies that $\cR_{\ell}^{-1}\parenv{\cR^{\prime\prime}}$ is obtained from $\tilde{\bmx}$ by deleting some complete minimum periods from certain periodic substrings of $\tilde{\bmx}$. This observation motivates the following definition.
\begin{definition}\label{def_checkpattern}
    Suppose that $\ell\ge 2$, $t\ge1$ and $0\le s<\min\mathset{\ell-1,t}$. Let $\bmz$ be a sequence over $\Sigma_q$, and let $\bmp$ be a periodic substring of $\bmz$ with minimum period $T$. We say that  $\bmp$ is an $(\ell,t,s)$-\emph{check pattern} in $\bmz$, if the following three conditions are satisfied:
    \begin{enumerate}[$(i)$]
        \item $T\le t-s$;
        \item $\bmp$ is a maximal periodic substring of $\bmz$ with minimum period $T$;
        \item $\abs{\bmp}\ge\ell-1-s$.
    \end{enumerate}
    We say that $\bmp$ is an $(\ell,t,\le s)$-check pattern, if it is an $(\ell,t,s^\prime)$-check pattern for some $s^\prime\le s$.
\end{definition}

Thus, the sequence $\cR_{\ell}^{-1}\parenv{\cR^{\prime\prime}}$ is obtained from $\tilde{\bmx}$ by deleting some complete minimum periods from some $(\ell,t,s)$-check patterns. In \Cref{lem_scheckpattern} and \Cref{cor_tcheckpattern} below, we show that these deletions in $\tilde{\bmx}$ only shorten these check patterns, without changing their number, their order, or their minimum periods.

The following lemma follows immediately from the well-known Fine-Wilf Theorem \cite[Theorem 1]{FineWilf1965}, and can also be found in \cite[Lemma 1]{kolpakov1999FOCS}.
\begin{lemma}\cite[Lemma 1]{kolpakov1999FOCS}\label{lem_periodoverlap}
    Let $\bmp_1$ and $\bmp_2$ be two maximal periodic substrings of a sequence $\bmz$, with minimum periods $T_1$ and $T_2$, respectively. Then $\bmp_1$ and $\bmp_2$ can overlap in at most $T_1+T_2-\gcd\parenv{T_1,T_2}-1$ positions.
\end{lemma}

The following corollary is an immediate consequence of \Cref{lem_periodoverlap}.
\begin{corollary}\label{cor_maxoverlap}
    Suppose that $s<t$ and $t\le\frac{\ell+1+s}{2}$. Let $\bmp_1$ and $\bmp_2$ be two distinct maximal periodic substrings of a sequence $\bmz$, with minimum periods $T_1$ and $T_2$, respectively. If $T_1,T_2\le t-s$, then $\bmp_1$ and $\bmp_2$ can overlap in at most $\ell-2-s$ positions.
\end{corollary}

\begin{lemma}\label{lem_scheckpattern}
    Let $\bmx\in\Sigma_q^n$ and $1\le t\le\min\mathset{\ell-2,\frac{\ell+1+s}{2}}$, where $0\le s<t$. Suppose that $\cR^\prime$ is obtained from $\cR_{\ell}(\tilde{\bmx})$ by deleting $t$ arbitrary $\ell$-mers, and that $\cR^{\prime\prime}$ is the consistent vector of $\ell$-mers obtained by inserting a minimum number $s$ of $\ell$-mers into $\cR^\prime$. Let $\bmy=\cR_{\ell}^{-1}\parenv{\cR^{\prime\prime}}$.
    \begin{enumerate}[$(i)$]
        \item The sequences $\tilde{\bmx}$ and $\bmy$ contain the same number of $(\ell,t,s)$-check patterns.
        \item Let $\bmp^\prime$ be the $i$-th $(\ell,t,s)$-check pattern in $\bmy$, and let $\bmp$ be the $i$-th $(\ell,t,s)$-check pattern in $\tilde{\bmx}$. Then $\bmp^\prime$ and $\bmp$ have the same minimum period. Moreover, either $\bmp^\prime=\bmp$, or $\bmp^\prime$ is obtained from $\bmp$ by deleting some complete minimum periods.
    \end{enumerate}
\end{lemma}
\begin{IEEEproof}
    Set $P\triangleq t-s$ and  $L\triangleq \ell-1-s$.
    Then $P\ge 1$ and, since $t\le \ell-2$, we have $L>P$. Moreover, the assumption
    $t\le (\ell+1+s)/2$ implies that \Cref{cor_maxoverlap} is applicable with this value of $s$. Hence, throughout the proof, any two distinct maximal periodic substrings whose minimum periods are at most $P$ can overlap in at most $\ell-2-s=L-1$
    positions.

    We shall use the following two elementary facts repeatedly.

    First, let $\bmz$ be an arbitrary sequence, and let $\bmw$ be a length-$L$ substring of $\bmz$ whose minimum period is $T\le P$. Extend $\bmw$ maximally in $\bmz$ while preserving period $T$, and denote the resulting substring by $\bmq$. Then $\bmq$ is an $(\ell,t,s)$-check pattern in $\bmz$. Indeed, $\bmq$ has period $T\le P=t-s$, is maximal with respect to this period, and has length at least $L=\ell-1-s$. Moreover, this check pattern is unique. To see this, suppose that $\bmw$ is contained in another $(\ell,t,s)$-check pattern $\bmq'$ whose minimum period is at most $P$. Then $\bmq$ and $\bmq'$ overlap in at least $L$ positions, contradicting \Cref{cor_maxoverlap} unless $\bmq=\bmq'$. In particular, every length-$L$ substring with minimum period at most $P$ belongs to a unique $(\ell,t,s)$-check pattern.

    Second, suppose that a periodic sequence $\bma=a_1a_2\cdots a_m$ has period $T$, and that a substring of length $h$ is deleted from $\bma$, where $T\mid h$. If the deleted substring starts after the first $r$ symbols of $\bma$, then the resulting sequence is $a_1\cdots a_r a_{r+h+1}\cdots a_m$.
    Since $T\mid h$ and $\bma$ has period $T$, we have $a_{j+h}=a_j$ whenever both sides are defined. Therefore the resulting sequence is exactly $a_1a_2\cdots a_{m-h}$.
    Thus, deleting any length-$h$ substring from a $T$-periodic sequence, with $T\mid h$, has the same effect as shortening this sequence by $h$ symbols, and the shortened sequence is still $T$-periodic.

    We next prove a local stability claim.

    \smallskip
    \noindent\emph{Claim.}
    Let $\bmz^\prime$ be obtained from a sequence $\bmz$ by deleting a substring of length $h$ from an $(\ell,t,s)$-check pattern $\bma$ of $\bmz$, where the minimum period of $\bma$ is $T$, $T\mid h$, and the shortened copy $\bma^\prime$ of $\bma$ has length at least $L$. Then the $(\ell,t,s)$-check patterns in $\bmz$ and $\bmz^\prime$ are in one-to-one order-preserving correspondence. Under this correspondence, $\bma$ is mapped to $\bma^\prime$, and every other check pattern is unchanged as a sequence. In particular, corresponding check patterns have the same minimum period.

    \smallskip
    \noindent\emph{Proof of the claim.}
    By the second fact above, $\bma^\prime$ is obtained from $\bma$ by deleting $h/T$ complete periods, and $\bma^\prime$ is still $T$-periodic.

    We first show that the minimum period of $\bma^\prime$ is still $T$. Suppose, to the contrary, that $\bma^\prime$ has minimum period $T'<T$. Since $\abs{\bma^\prime}\ge L$, the length-$L$ prefix of $\bma^\prime$ has minimum period at most $T'$. This prefix is also a length-$L$ substring of $\bma$, because deleting complete periods from $\bma$ preserves its length-$L$ prefix as a word. By the first fact, this length-$L$ substring is contained in an $(\ell,t,s)$-check pattern of $\bmz$ with minimum period at most $T'$. This check pattern is distinct from $\bma$ and overlaps $\bma$ in at least $L$ positions, contradicting \Cref{cor_maxoverlap}. Hence the minimum period of $\bma^\prime$ is $T$.

    We now prove that $\bma^\prime$ is maximal $T$-periodic in $\bmz^\prime$. Since $T\mid h$, the first $T$ symbols and the last $T$ symbols of $\bma^\prime$ are in the same phase as the first $T$ symbols and the last $T$ symbols of $\bma$, respectively. Therefore, the symbol immediately to the left of $\bma^\prime$, if it exists, can extend $\bma^\prime$ with period $T$ if and only if the symbol immediately to the left of $\bma$ can extend $\bma$ with period $T$. The same is true for the symbol immediately to the right. Since $\bma$ is maximal $T$-periodic in $\bmz$, neither extension is possible. Hence $\bma^\prime$ is maximal $T$-periodic in $\bmz^\prime$. Thus $\bma^\prime$ is an $(\ell,t,s)$-check pattern in $\bmz^\prime$.

    Let $\bmb$ be any $(\ell,t,s)$-check pattern in $\bmz$ distinct from $\bma$, and let $U$ be its minimum period. By \Cref{cor_maxoverlap}, $\bma$ and $\bmb$ overlap in at most $L-1$ positions. Since $\abs{\bma^\prime}\ge L$, we have $\abs{\bma}-h=\abs{\bma^\prime}\ge L$.
    Hence
    $\abs{\bma\setminus\bmb}
        \ge \abs{\bma}-(L-1)
        \ge h+1$,
    where $\bma\setminus\bmb$ denotes the substring of $\bma$ that is outside of the overlap of $\bma$ and $\bmb$.
    Since $\bma$ and $\bmb$ are substrings, their intersection is an interval, possibly empty. Therefore $\bma$ contains a length-$h$ substring disjoint from $\bmb$. By the second fact above, deleting this disjoint length-$h$ substring from $\bma$ gives the same resulting sequence $\bmz^\prime$. Thus we may regard the deletion as being performed outside $\bmb$. Under this equivalent realization, the occurrence of $\bmb$ is not changed.

    It remains only to check that $\bmb$ is still maximal $U$-periodic in $\bmz^\prime$. The only possible change in the maximality of $\bmb$ could occur at the side facing $\bma$. However, the symbol adjacent to $\bmb$ on that side is either unchanged, or is replaced by another symbol of $\bma$ at distance $h$ from the original one. Since $T\mid h$ and $\bma$ is $T$-periodic, these two adjacent symbols are equal. Hence the same boundary obstruction that prevented $\bmb$ from being extended with period $U$ in $\bmz$ also prevents such an extension in $\bmz^\prime$. Therefore $\bmb$ remains a maximal $U$-periodic substring in $\bmz^\prime$, and so it remains an $(\ell,t,s)$-check pattern with the same minimum period.

    Conversely, let $\bmc^\prime$ be an $(\ell,t,s)$-check pattern in $\bmz^\prime$ distinct from $\bma^\prime$. We reverse the above deletion by inserting into $\bma^\prime$ a length-$h$ block that consists of complete periods of length $T$. Since inserting such a block at any position of the $T$-periodic sequence $\bma^\prime$ gives the same $T$-periodic sequence $\bma$, we may choose the insertion position outside the overlap between $\bma^\prime$ and $\bmc^\prime$. This is possible because two distinct check patterns in $\bmz^\prime$ overlap in at most $L-1$ positions, whereas $\abs{\bma^\prime}\ge L$. With this choice, $\bmc^\prime$ is not changed as a sequence. The same boundary argument as above shows that its maximality and minimum period are preserved after the insertion. Hence $\bmc^\prime$ comes from an $(\ell,t,s)$-check pattern of $\bmz$.

    Thus no check pattern is lost, no new check pattern is created, and no two check patterns merge. The correspondence is order-preserving because deleting a substring does not change the left-to-right order of the remaining substrings. This proves the claim.

    \smallskip

    We now apply the claim to the passage from $\tilde{\bmx}$ to $\bmy$. By \Cref{cor_arbitrarydel} (ii)--(v), the sequence $\bmy$ is obtained from $\tilde{\bmx}$ by deleting a collection of bursts whose total length is $P=t-s$. Moreover, each such burst is contained in a maximal periodic substring whose minimum period is at most $P$, and, since $s<t\le \ell-2$, \Cref{cor_arbitrarydel} (iv) implies that the length of each deleted burst is divisible by the minimum period of the corresponding maximal periodic substring. Finally, \Cref{cor_arbitrarydel} (v) guarantees that after all deletions from any such substring, the remaining sequence has length at least $\ell-1-s=L$.
    Therefore, each burst-deletion satisfies the hypothesis of the claim, possibly after grouping consecutive complete-period deletions inside the same maximal periodic substring.

    Applying the claim successively to all these burst-deletions, we obtain an order-preserving bijection between the set of $(\ell,t,s)$-check patterns in $\tilde{\bmx}$ and the set of $(\ell,t,s)$-check patterns in $\bmy$. This proves (i).

    It remains to verify the structural statement in (ii). Let $\bmp$ be the $i$-th $(\ell,t,s)$-check pattern in $\tilde{\bmx}$, and let $\bmp^\prime$ be the corresponding, equivalently the $i$-th, $(\ell,t,s)$-check pattern in $\bmy$. At each local deletion step, the claim shows that the corresponding check pattern keeps the same minimum period, and that the only possible change is the deletion of complete minimum periods from that check pattern. Hence $\bmp^\prime$ and $\bmp$ have the same minimum period. Moreover, if no complete minimum period is deleted from $\bmp$, then $\bmp^\prime=\bmp$; otherwise, $\bmp^\prime$ is obtained from $\bmp$ by deleting some complete minimum periods. This proves (ii) and completes the proof.
\end{IEEEproof}

Let $s$ be as above. In fact, as the following corollary demonstrates, the deletions in $\tilde{\bmx}$ do not affect any $(\ell,t,\le t-1)$-check pattern that is not an $(\ell,t,s)$-check pattern.
\begin{corollary}\label{cor_tcheckpattern}
    Let $\bmx\in\Sigma_q^n$ and $1\le t\le\min\mathset{\ell-2,\frac{\ell}{2}}$. Suppose that $\cR^\prime$ is obtained from $\cR_{\ell}(\tilde{\bmx})$ by deleting $t$ arbitrary $\ell$-mers, and that $\cR^{\prime\prime}$ is the consistent vector of $\ell$-mers obtained by inserting a minimum number $s$ of $\ell$-mers into $\cR^\prime$, where $0\le s<t$. Let $\bmy=\cR_{\ell}^{-1}\parenv{\cR^{\prime\prime}}$.
    \begin{enumerate}[$(i)$]
        \item The sequences $\tilde{\bmx}$ and $\bmy$ contain the same number of $(\ell,t,\le t-1)$-check patterns.
        \item Let $\bmp^\prime$ be the $i$-th $(\ell,t,\le t-1)$-check pattern in $\bmy$, and let $\bmp$ be the $i$-th $(\ell,t,\le t-1)$-check pattern in $\tilde{\bmx}$. Then $\bmp^\prime$ and $\bmp$ have the same minimum period. If $\bmp$ is not an $\parenv{\ell,t,s}$-check pattern, it must be that $\bmp^\prime=\bmp$. If $\bmp$ is an $\parenv{\ell,t,s}$-check pattern, then either $\bmp^\prime=\bmp$, or $\bmp^\prime$ is obtained from $\bmp$ by deleting some complete minimum periods.
    \end{enumerate}
\end{corollary}
\begin{IEEEproof}
    Since $t\le \ell/2$, we have
    $t\le \frac{\ell+1+s}{2}$.
    Hence, \Cref{lem_scheckpattern} applies to the actual value of $s$.
    By \Cref{cor_arbitrarydel}, the sequence $\bmy$ is obtained from
    $\tilde{\bmx}$ by deleting some complete minimum periods from certain
    $(\ell,t,s)$-check patterns. By \Cref{lem_scheckpattern}, the
    $(\ell,t,s)$-check patterns in $\tilde{\bmx}$ and $\bmy$ are in
    one-to-one order-preserving correspondence, and corresponding patterns
    have the same minimum period. Moreover, a corresponding pair is either
    identical, or the latter is obtained from the former by deleting some
    complete minimum periods.

    It remains to show that no $(\ell,t,\le t-1)$-check pattern which is
    not an $(\ell,t,s)$-check pattern is affected. Let $\bmp$ be an
    $(\ell,t,r)$-check pattern in $\tilde{\bmx}$ for some
    $0\le r\le t-1$, and suppose that $\bmp$ is not an
    $(\ell,t,s)$-check pattern. Let $\bmq$ be an $(\ell,t,s)$-check
    pattern from which symbols are deleted. Denote by $U$ and $T$ the
    minimum periods of $\bmp$ and $\bmq$, respectively. Then $U\le t-r$ and $T\le t-s$.
    If $\bmp\ne\bmq$, then by \Cref{lem_periodoverlap},
    $$
        \abs{\bmp\cap\bmq}
        \le U+T-\gcd(U,T)-1
        \le (t-r)+(t-s)-2
        \le \ell-r-s-2,
    $$
    where $\bmp\cap\bmq$ denotes the overlap of $\bmp$ and $\bmq$.
    In particular, $\abs{\bmp\cap\bmq}<\ell-1-r\le \abs{\bmp}$
    and $ \abs{\bmp\cap\bmq}\le \ell-2-s$.
    Thus $\bmp$ is not contained in $\bmq$, and $\bmq$ is not contained in
    $\bmp$.

    Let $H$ be the total number of symbols deleted from $\bmq$. By
    \Cref{cor_arbitrarydel} (v), after these deletions the remaining part
    of $\bmq$ has length at least $\ell-1-s$. Hence, $H\le \abs{\bmq}-(\ell-1-s)$.
    Since $\abs{\bmp\cap\bmq}\le \ell-2-s$, we get $H\le \abs{\bmq}-\abs{\bmp\cap\bmq}-1$.
    Therefore, the part of $\bmq$ outside $\bmp$ contains a substring of
    length at least $H$.

    Since the deletions from $\bmq$ are deletions of complete minimum
    periods, we have $T\mid H$. Deleting $H$ consecutive symbols from a
    $T$-periodic sequence, with $T\mid H$, gives the same shortened
    sequence regardless of the deletion position. Hence the deletions from
    $\bmq$ may be equivalently realized inside $\bmq$ but outside $\bmp$.
    Under this equivalent realization, $\bmp$ is not changed as a sequence.
    Moreover, any boundary symbol of $\bmp$ which is shifted through
    $\bmq$ is replaced by a symbol in the same $T$-phase, and hence by the
    same symbol. Consequently, the maximality and the minimum period of
    $\bmp$ are preserved.

    Thus every $(\ell,t,\le t-1)$-check pattern in $\tilde{\bmx}$ which
    is not an $(\ell,t,s)$-check pattern appears unchanged in $\bmy$.
    Conversely, applying the same argument to the inverse operation, namely
    reinserting the deleted complete minimum periods into the corresponding
    $(\ell,t,s)$-check patterns of $\bmy$, shows that every
    $(\ell,t,\le t-1)$-check pattern in $\bmy$ arises from one in
    $\tilde{\bmx}$.

    Combining this with the correspondence for $(\ell,t,s)$-check patterns
    given by \Cref{lem_scheckpattern}, we obtain an order-preserving
    bijection between the $(\ell,t,\le t-1)$-check patterns in
    $\tilde{\bmx}$ and those in $\bmy$. Therefore the two sequences contain
    the same number of $(\ell,t,\le t-1)$-check patterns. Furthermore,
    corresponding check patterns have the same minimum period. If the
    corresponding pattern in $\tilde{\bmx}$ is not an $(\ell,t,s)$-check
    pattern, it is unchanged; if it is an $(\ell,t,s)$-check pattern, then
    by \Cref{lem_scheckpattern} it is either unchanged or shortened by
    deleting complete minimum periods. This proves both claims.
\end{IEEEproof}

The preceding analysis shows that, to recover $\tilde{\bmx}$, it suffices to determine which check patterns suffered deletions and the number of deletions each of them suffered, since then we can extend the corresponding check patterns in $\bmy$ according to their minimum periods. Our constructions are presented in \Cref{subsec_1del,subsec_multidel}.

\subsection{Correcting Single Deletion}\label{subsec_1del}
Given a sequence $\bmx\in\Sigma_q^n$, let $m=m(\tilde{\bmx})$ be the number of $(\ell,1,0)$-check patterns in $\tilde{\bmx}$. Since any two different $(\ell,1,0)$-check patterns do not overlap and $\abs{\tilde{\bmx}}=n+2\ell$, we have $m\le\floorenv{(n+2\ell)/(\ell-1)}$. If $m=0$, set $f_1(\tilde{\bmx},i)=0$ for all $-\ell+1\le i\le n+\ell$.

If $m\ge 1$, let $\tilde{\bmx}_{[i_{j,1},i_{j,2}]}$ be the $j$-th $(\ell,1,0)$-check pattern in $\tilde{\bmx}$, where $i_{j,2}<i_{j+1,1}$ for $1\le j< m$. For all $-\ell+1\le i\le n+\ell$, set
\begin{equation*}
    f_1(\tilde{\bmx},i)=
    \begin{cases}
        j,\mbox{ if }i\in[i_{j,1},i_{j,2}],\\
        0,\mbox{ else}.
    \end{cases}
\end{equation*}
Define $F_1\parenv{\tilde{\bmx}}\triangleq\sum_{i=-\ell+1}^{n+\ell}f_1(\tilde{\bmx},i)$.

\begin{theorem}\label{thm_1del}
   Let $n>\ell\ge2$. For a given integer $0\le a<\floorenv{(n+2\ell)/(\ell-1)}$, define the code
    \begin{equation*}
        \cC(n;a)=\mathset{\bmx\in\Sigma_q^n:F_1\parenv{\tilde{\bmx}}\equiv a\pmod{\floorenv{(n+2\ell)/(\ell-1)}}}.
    \end{equation*}
    This code can correct a single deletion in $\cR_{\ell}(\tilde{\bmx})$. There exists some $a$ such that $\rho\parenv{\cC(n;a)}\le\log\parenv{\floorenv{(n+2\ell)/(\ell-1)}}$.
\end{theorem}
\begin{IEEEproof}
    Suppose that $\tilde{\bmx}=\bmu\bmx\bmv$ is transmitted, and let $\cR^\prime$ be the received vector of $\ell$-mers, where $\bmx\in\cC(n;a)$ and $\cR^\prime$ is a subvector of $\cR_{\ell}(\tilde{\bmx})$ of length $n+\ell$. Our goal is to recover $\tilde{\bmx}$ from $\cR^\prime$.

    By \Cref{lem_burstdel}, we can insert $s\in\{0,1\}$ $\ell$-mers into $\cR_{\ell}^{\prime}$ to obtain a consistent vector of $\ell$-mers $\cR^{\prime\prime}$. The value of $s$ is known to the decoder. If $s=1$, then \Cref{lem_burstdel} (i) gives $\tilde{\bmx}=\cR_{\ell}^{-1}\parenv{\cR^{\prime\prime}}$.

    Now assume $s=0$. Let $\bmy=\cR_{\ell}^{-1}\parenv{\cR^{\prime\prime}}$. By \Cref{lem_burstdel}, the sequence $\bmy$ is obtained from $\tilde{\bmx}$ by deleting one symbol from an $(\ell,1,0)$-check pattern of length at least $\ell$ ($\ge2$). Thus $m(\tilde{\bmx})\ge1$. Moreover, since the deletion occurred in an $(\ell,1,0)$-check pattern of length at least $\ell$ and any two distinct such patterns do not overlap, we can conclude that $m(\tilde{\bmx})=m(\bmy)$ and that the $i$-th $(\ell,1,0)$-check pattern in $\bmy$ is obtained from the $i$-th $(\ell,1,0)$-check pattern in $\tilde{\bmx}$.

    Assume that the deletion occurred in the $j$-th $(\ell,1,0)$-check pattern of $\tilde{\bmx}$. Then $j=F_1\parenv{\tilde{\bmx}}-F_1(\bmy)\equiv a-F_1(\bmy)\pmod{\floorenv{(n+2\ell)/(\ell-1)}}$. Since $j$ is the unique integer in $\sparenv{1,\floorenv{(n+2\ell)/(\ell-1)}}$ satisfying this congruence, the decoder can determine $j$. The original sequence $\tilde{\bmx}$ is then recovered by extending the $j$-th $(\ell,1,0)$-check pattern in $\bmy$ by one symbol. This completes the proof.
\end{IEEEproof}

\subsection{Correcting Multiple Deletions}\label{subsec_multidel}
In this subsection, we first construct $\ell$-read $t$-deletion correcting codes for the case $2\le t\le\ell/2$. The resulting redundancy is $t\log n+O(1)$. We then consider the case $\ell=2t-1$, where $t\ge 3$, and construct a code with redundancy $(2t-1)\log n+O(1)$.

The following lemma is a key ingredient in our constructions.
\begin{lemma}\cite[Proof of Lemma 4.1]{Lara2010SIAM}\label{lem_pre1}
  Let $Q>1$ be an integer and $p\ge\max\{Q,N+1\}$ be a prime number. Suppose that integers $1\le\alpha_1,\ldots,\alpha_N\le Q$ satisfy the system of congruences
  \begin{equation}\label{eq_pre1}
    \left\{
    \begin{array}{c}
     \alpha_1+\cdots+\alpha_N\equiv a_1\pmod{p}\\
     \alpha_1^2+\cdots+\alpha_N^2\equiv a_2\pmod{p}\\
     \vdots\\
     \alpha_1^N+\cdots+\alpha_N^N\equiv a_N\pmod{p}
    \end{array}
    \right.,
  \end{equation}
where $a_1,\ldots,a_N\in\mathbb{Z}_p$ are given. Then the multiset $\multiset{\alpha_1,\ldots,\alpha_N}$ can be efficiently and uniquely determined by solving \eqref{eq_pre1}.
\end{lemma}

Assume that $n>\ell$ and $2\le t\le\min\mathset{\ell-2,\frac{\ell}{2}}$. Given a sequence $\bmx\in\Sigma_q^n$, let $m=m(\tilde{\bmx})$ be the total number of $(\ell,t,\le t-1)$-check patterns in $\tilde{\bmx}$. For any $s\le t-1$, an $(\ell,t,s)$-check pattern has length at least $\ell-1-s$ and can overlap with another $(\ell,t,\le t-1)$-check pattern in at most $2t-s-2\le \ell-2-s$ positions. Hence no $(\ell,t,\le t-1)$-check pattern can be contained in another one. Consequently, we have $m\le n+\ell+t+1$.

If $m=0$, set $f_t(\tilde{\bmx},i)=0$ for all $-\ell+1\le i\le n+\ell$.
If $m\ge 1$, for $1\le j\le m$, let $\tilde{\bmx}_{[i_{j,1},i_{j,2}]}$ be the $j$-th $(\ell,t,\le t-1)$-check pattern in $\tilde{\bmx}$. Since no $(\ell,t,\le t-1)$-check pattern is contained in another, we have $i_{j,1}<i_{j+1,1}$ and $i_{j,2}<i_{j+1,2}$ for every $1\le j< m$. This implies that $[i_{j,1},i_{j,2}]\setminus[i_{j-1,1},i_{j-1,2}]\ne\emptyset$ for all $j$.

For all $-\ell+1\le i\le n+\ell$, set
\begin{equation*}
    f_t(\tilde{\bmx},i)=
    \begin{cases}
        j,\mbox{ if }i\in[i_{j,1},i_{j,2}]\setminus[i_{j-1,1},i_{j-1,2}],\\
        0,\mbox{ else}.
    \end{cases}
\end{equation*}
Define $F_{r}\parenv{\tilde{\bmx}}\triangleq\sum_{i=-\ell+1}^{n+\ell}f_t(\tilde{\bmx},i)^r$ for all $r\in[t]$.

\begin{theorem}\label{thm_multidel1}
   Let $n>\ell\ge4$ and $2\le t\le\frac{\ell}{2}$ (so $t\le\ell-2$ holds). Let $p\ge\max\mathset{n+\ell+t,t+1}$ be the smallest prime. For given integers $0\le a_1,\ldots,a_t<p$, set $\bma=\parenv{a_1,\ldots,a_t}$ and define the code
    \begin{equation*}
        \cC(n;\bma)=\mathset{\bmx\in\Sigma_q^n:F_{r}\parenv{\tilde{\bmx}}\equiv a_r\pmod{p},\forall r\in[t]}.
    \end{equation*}
    This code can correct $t$ deletions in $\cR_{\ell}(\tilde{\bmx})$. Moreover, there exists some $\bma$ such that $\rho\parenv{\cC(n;\bma)}\le t\log\parenv{p}=t\log(n)+O(1)$.
\end{theorem}
\begin{IEEEproof}
Suppose that $\tilde{\bmx}=\bmu\bmx\bmv$ is transmitted, and that $\cR^\prime$ is the received vector of $\ell$-mers, where $\bmx\in\cC(n;\bma)$ and $\cR^\prime$ is a subvector of $\cR_{\ell}(\tilde{\bmx})$ of length $n+\ell+1-t$. We aim to recover $\tilde{\bmx}$ from $\cR^\prime$.

By \Cref{cor_arbitrarydel}, we can insert $s\in\{0,1,\ldots,t\}$ $\ell$-mers into $\cR_{\ell}^{\prime}$ to obtain a consistent vector of $\ell$-mers $\cR^{\prime\prime}$. The value of $s$ is known to us. Let $\bmy=\cR_{\ell}^{-1}\parenv{\cR^{\prime\prime}}$. If $s=t$, then \Cref{cor_arbitrarydel} (ii) gives $\tilde{\bmx}=\bmy$.

Now suppose that $s<t$. By the same corollary, the sequence $\bmy$ is obtained from $\tilde{\bmx}$ by deleting some complete minimum periods from certain $(\ell,t,\le t-1)$-check patterns, and the total number of deleted symbols is $t-s$. In particular, $m(\tilde{\bmx})\ge1$. By \Cref{cor_tcheckpattern}, we have $m(\tilde{\bmx})=m(\bmy)$. In addition, the deletions preserve both the order and the minimum period of every $(\ell,t,\le t-1)$-check pattern.

For $1\le i\le t-s$, assume that the $i$-th deletion in $\tilde{\bmx}$ occurred in the $j_i$-th $(\ell,t,\le t-1)$-check pattern, where $1\le j_1\le\cdots\le j_{t-s}\le m(\bmy)$. Then it holds that $\sum_{i=1}^{t-s}j_i^r=F_{r}\parenv{\tilde{\bmx}}-F_{r}(\bmy)\equiv a_r-F_{r}(\bmy)\pmod{p}$ for all $1\le r\le t$. Let $\delta_r\in\mathset{0,1,\ldots,p-1}$ be the unique integer such that $\delta_r\equiv a_r-F_{r}(\bmy)\pmod{p}$. Then the unknown indices $j_1,\ldots,j_{t-s}$ satisfy the system
\begin{equation}\label{eq_pre2}
    \left\{
    \begin{array}{c}
        j_1+\cdots+j_{t-s}\equiv \delta_1\pmod{p},\\
        j_1^2+\cdots+j_{t-s}^2\equiv \delta_2\pmod{p},\\
        \vdots\\
        j_1^{t-s}+\cdots+j_{t-s}^{t-s}\equiv \delta_{t-s}\pmod{p}.
    \end{array}
    \right.
\end{equation}
By \Cref{lem_pre1}, the multiset $\multiset{j_1,\ldots,j_{t-s}}$ can be efficiently and uniquely determined by solving \eqref{eq_pre2}. Since $j_1\le\cdots\le j_{t-s}$, the tuple $\parenv{j_1,\ldots,j_{t-s}}$ is uniquely determined. In other words, we have identified which $(\ell,t,\le t-1)$-check patterns suffered deletions and exactly how many deletions each of them suffered. Now $\tilde{\bmx}$ can be recovered by extending the corresponding $(\ell,t,\le t-1)$-check patterns in $\bmy$ according to their minimum periods. This completes the proof.
\end{IEEEproof}

The construction above requires $\ell\ge 2t$. We next present a construction for the case $\ell=2t-1$. Let $\cR^\prime$ and $\cR^{\prime\prime}$ be as in the proof of \Cref{thm_multidel1}. In particular, $\cR^{\prime\prime}$ is obtained from $\cR^\prime$ by inserting $s$ $\ell$-mers, where $0\le s\le t$. In the previous proof, all cases $0\le s\le t-1$ were treated uniformly. However, when $s=0$, by \Cref{lem_scheckpattern}, it is sufficient to require that $t\le\frac{\ell+1}{2}$. For the remaining cases $1\le s<t$, we use the following lemma, where an $(\ell,t,[1,t-1])$-check pattern means an $(\ell,t,s^\prime)$-check pattern for some $s^\prime\in[1,t-1]$.
\begin{lemma}\label{lem_tcheckpattern}
    Let $\bmx\in\Sigma_q^n$ and $1\le t\le\min\mathset{\ell-2,\frac{\ell+1}{2}}$. Suppose that $\cR^\prime$ is obtained from $\cR_{\ell}(\tilde{\bmx})$ by deleting $t$ arbitrary $\ell$-mers, and let $\cR^{\prime\prime}$ be the consistent vector of $\ell$-mers obtained by inserting a minimum number $s$ of $\ell$-mers into $\cR^\prime$, where $1\le s<t$. Let $\bmy=\cR_{\ell}^{-1}\parenv{\cR^{\prime\prime}}$. Then the following statements hold.
    \begin{enumerate}[$(i)$]
        \item The sequences $\tilde{\bmx}$ and $\bmy$ contain the same number of $(\ell,t,[1,t-1])$-check patterns.
        \item Let $\bmp^\prime$ be the $i$-th $(\ell,t,[1,t-1])$-check pattern in $\bmy$, and $\bmp$ be the $i$-th $(\ell,t,[1,t-1])$-check pattern in $\tilde{\bmx}$. Then $\bmp^\prime$ and $\bmp$ have the same minimum period. If $\bmp$ is not an $\parenv{\ell,t,s}$-check pattern, it must be that $\bmp^\prime=\bmp$. If $\bmp$ is an $\parenv{\ell,t,s}$-check pattern, then either $\bmp^\prime=\bmp$, or $\bmp^\prime$ is obtained from $\bmp$ by deleting some complete minimum periods.
    \end{enumerate}
\end{lemma}

The proof is identical to that of \Cref{cor_tcheckpattern}, except for the following observation. When $1\le s<t$, an $(\ell,t,[1,t-1])$-check pattern and a distinct $(\ell,t,s)$-check pattern can overlap in at most $t-1+t-s-2\le \ell-3-s$ positions (when the two check patterns are distinct). Thus, the weaker condition $t\le\frac{\ell+1}{2}$ is sufficient.

We now describe the construction for the case $\ell=2t-1$. The main idea is to handle the cases $s=0$ and $1\le s<t$ separately. For a given $\bmx\in\Sigma_q^n$, let $m_1=m_1(\tilde{\bmx})$ be the total number of $(\ell,t,0)$-check patterns in $\tilde{\bmx}$, and let $m_2=m_2(\tilde{\bmx})$ be the total number of $(\ell,t,[1,t-1])$-check patterns in $\tilde{\bmx}$. Since the length of an $(\ell,t,0)$-check pattern is at least $\ell-1$ and two such check patterns overlap in at most $\ell-2$ positions, it holds that $m_1\le n+\ell+2$. Similarly, we have $m_2\le n+\ell+t+1$.

For the case $s=0$, we introduce $t$ functions $G_{1,1},\ldots,G_{1,t}$. If $m_1=0$, set $g_1(\tilde{\bmx},i)=0$ for all $-\ell+1\le i\le n+\ell$.
If $m_1\ge 1$, for $1\le j\le m_1$, let $\tilde{\bmx}_{[i_{j,1},i_{j,2}]}$ be the $j$-th $(\ell,t,0)$-check pattern in $\tilde{\bmx}$, where $i_{j,1}<i_{j+1,1}$ and $i_{j,2}<i_{j+1,2}$ for $1\le j< m_1$. For all $-\ell+1\le i\le n+\ell$, set
\begin{equation*}
    g_1(\tilde{\bmx},i)=
    \begin{cases}
        j,\mbox{ if }i\in[i_{j,1},i_{j,2}]\setminus[i_{j-1,1},i_{j-1,2}],\\
        0,\mbox{ else}.
    \end{cases}
\end{equation*}
Define $G_{1,r}\parenv{\tilde{\bmx}}\triangleq\sum_{i=-\ell+1}^{n+\ell}g_1(\tilde{\bmx},i)^r$ for all $r\in[t]$.

For the case $1\le s<t$, we need $t-1$ functions $G_{2,1},\ldots,G_{2,t-1}$. If $m_2=0$, define $g_2(\tilde{\bmx},i)=0$ for all $-\ell+1\le i\le n+\ell$.
If $m_2\ge 1$, for $1\le j\le m_2$, let $\tilde{\bmx}_{[i_{j,1},i_{j,2}]}$ be the $j$-th $(\ell,t,[1,t-1])$-check pattern in $\tilde{\bmx}$, where $i_{j,1}<i_{j+1,1}$ and $i_{j,2}<i_{j+1,2}$ for $1\le j< m_2$. For all $-\ell+1\le i\le n+\ell$, set
\begin{equation*}
    g_2(\tilde{\bmx},i)=
    \begin{cases}
        j,\mbox{ if }i\in[i_{j,1},i_{j,2}]\setminus[i_{j-1,1},i_{j-1,2}],\\
        0,\mbox{ else}.
    \end{cases}
\end{equation*}
Define $G_{2,r}\parenv{\tilde{\bmx}}\triangleq\sum_{i=-\ell+1}^{n+\ell}g_2(\tilde{\bmx},i)^r$ for all $r\in[t-1]$.

\begin{theorem}\label{thm_multidel2}
   Let $n>t\ge3$ and $\ell=2t-1$ (so $t\le\ell-2$). Let $p_1$ be the smallest prime greater than $n+\ell+1$ and $p_2$ be the smallest prime greater than $n+\ell+t+1$. For given integers $0\le a_1,\ldots,a_{t}<p_1$ and $0\le b_1,\ldots,b_{t-1}<p_2$, set $\bma=\parenv{a_1,\ldots,a_t}$, $\bmb=\parenv{b_1,\ldots,b_{t-1}}$. Define the code
    \begin{equation*}
        \cC(n;\bma,\bmb)=\mathset{\bmx\in\Sigma_q^n:
        \begin{array}{c}
             G_{1,r}\parenv{\tilde{\bmx}}\equiv a_r\pmod{p_1},\forall r\in[t],\\
             G_{2,r}(\tilde{\bmx})\equiv b_r\pmod{p_2}, \forall r\in[t-1]
        \end{array}
        }.
    \end{equation*}
    This code can correct $t$ deletions in $\cR_{\ell}(\tilde{\bmx})$. There are choices of $\bma$ and $\bmb$ such that $\rho\parenv{\cC(n;\bma,\bmb)}\le t\log\parenv{p_1}+(t-1)\log(p_2)=(2t-1)\log(n)+O(1)$.
\end{theorem}
\begin{IEEEproof}
Suppose that $\tilde{\bmx}=\bmu\bmx\bmv$ is transmitted, and that $\cR^\prime$ is the received vector of $\ell$-mers, where $\bmx\in\cC(n;\bma,\bmb)$ and $\cR^\prime$ is a subvector of $\cR_{\ell}(\tilde{\bmx})$ of length $n+\ell+1-t$. We aim to recover $\tilde{\bmx}$ from $\cR^\prime$.

By \Cref{cor_arbitrarydel}, we can insert $s\in\{0,1,\ldots,t\}$ $\ell$-mers into $\cR^{\prime}$ to obtain a consistent vector of $\ell$-mers $\cR^{\prime\prime}$. The value of $s$ is known to us. Let $\bmy=\cR_{\ell}^{-1}\parenv{\cR^{\prime\prime}}$. If $s=t$, by \Cref{cor_arbitrarydel} (ii), we have $\tilde{\bmx}=\bmy$.

Consider first the case $s=0$. By \Cref{cor_arbitrarydel}, the sequence $\bmy$ is obtained from $\tilde{\bmx}$ by deleting a total of $t$ symbols from $(\ell,t,0)$-check patterns in $\tilde{\bmx}$. In particular, $m_1(\tilde{\bmx})\ge1$. By \Cref{lem_scheckpattern}, we have $m_1(\tilde{\bmx})=m_1(\bmy)$. In addition, \Cref{cor_tcheckpattern} ensures that the deletions in $\tilde{\bmx}$ do not change the order of $(\ell,t,0)$-check patterns, nor the minimum period of any such check pattern.

For $1\le i\le t$, assume that the $i$-th deletion in $\tilde{\bmx}$ occurred in the $j_i$-th $(\ell,t,0)$-check pattern, where $1\le j_1\le\cdots\le j_{t}\le m_1(\bmy)$. Then we have $\sum_{i=1}^{t}j_i^r=G_{1,r}\parenv{\tilde{\bmx}}-G_{1,r}(\bmy)\equiv a_r-G_{1,r}(\bmy)\pmod{p_1}$ for all $1\le r\le t$. Let $\delta_r\in\mathset{0,1,\ldots,p-1}$ be the unique integer such that $\delta_r\equiv a_r-G_{1,r}(\bmy)\pmod{p_1}$. Then it follows that
\begin{equation}\label{eq_pre3}
    \left\{
    \begin{array}{c}
        j_1+\cdots+j_{t}\equiv \delta_1\pmod{p_1},\\
        j_1^2+\cdots+j_{t}^2\equiv \delta_2\pmod{p_1},\\
        \vdots\\
        j_1^{t-s}+\cdots+j_{t}^{t}\equiv \delta_{t}\pmod{p_1}.
    \end{array}
    \right.
\end{equation}
By \Cref{lem_pre1}, the multiset $\multiset{j_1,\ldots,j_{t}}$ can be uniquely determined by \eqref{eq_pre3}. Since $j_1\le\cdots\le j_{t}$, the tuple $\parenv{j_1,\ldots,j_{t}}$ is uniquely determined. In other words, we have identified which $(\ell,t,0)$-check patterns suffered deletions and exactly how many deletions each of them suffered. Now $\tilde{\bmx}$ can be recovered by extending the corresponding $(\ell,t,0)$-check patterns in $\bmy$ according to their minimum periods.

It remains to consider the case $1\le s<t$. By \Cref{cor_arbitrarydel}, the sequence $\bmy$ is obtained from $\tilde{\bmx}$ by deleting a total of $t-s$ symbols from $(\ell,t,[1,t-1])$-check patterns in $\tilde{\bmx}$. In particular, $m_2(\tilde{\bmx})\ge1$. By \Cref{lem_tcheckpattern}, we have $m_2(\tilde{\bmx})=m_2(\bmy)$. Moreover, this lemma guarantees that the deletions preserve both the order and the minimum period of every $(\ell,t,[1,t-1])$-check pattern.

For $1\le i\le t-s$, assume that the $i$-th deletion in $\tilde{\bmx}$ occurred in the $j_i$-th $(\ell,t,[1,t-1])$-check pattern, where $1\le j_1\le\cdots\le j_{t-s}\le m_2(\bmy)$. Then it holds that $\sum_{i=1}^{t-s}j_i^r=G_{2,r}\parenv{\tilde{\bmx}}-G_{2,r}(\bmy)\equiv b_r-G_{2,r}(\bmy)\pmod{p_2}$ for all $1\le r< t$. Let $\delta_r\in\mathset{0,1,\ldots,p-1}$ be the unique integer such that $\delta_r\equiv b_r-G_{2,r}(\bmy)\pmod{p_2}$. Then we have
\begin{equation}\label{eq_pre4}
    \left\{
    \begin{array}{c}
        j_1+\cdots+j_{t-s}\equiv \delta_1\pmod{p_2},\\
        j_1^2+\cdots+j_{t-s}^2\equiv \delta_2\pmod{p_2},\\
        \vdots\\
        j_1^{t-s}+\cdots+j_{t-s}^{t-s}\equiv \delta_{t-s}\pmod{p_2}.
    \end{array}
    \right.
\end{equation}
By \Cref{lem_pre1}, the multiset $\multiset{j_1,\ldots,j_{t-s}}$ can be uniquely determined by \eqref{eq_pre4}. Since $j_1\le\cdots\le j_{t-s}$, the tuple $\parenv{j_1,\ldots,j_{t-s}}$ is uniquely determined. In other words, we have identified which $(\ell,t,[1,t-1])$-check patterns suffered deletions and the number of deletions each of these $(\ell,t,[1,t-1])$-check patterns suffered. Now $\tilde{\bmx}$ can be recovered by extending the corresponding $(\ell,t,[1,t-1])$-check patterns in $\bmy$ according to their minimum periods. This completes the proof.
\end{IEEEproof}

\section{Sporadic Constructions}\label{sec_sporadic}
In this section, we focus on $\ell$-read $t$-deletion correcting codes for the parameter pairs $(\ell,t)\in\mathset{(2,2),(3,2),(3,3)}$. More specifically, for each $\ell\in\{2,3\}$, we construct binary $\ell$-read $2$-deletion correcting codes with redundancy $2\log(n)+O(1)$. We also construct a non-binary $2$-read $2$-deletion correcting code with redundancy $4\log(n)+O(1)$, a non-binary $3$-read $2$-deletion correcting code with redundancy $3\log(n)+O(1)$, a binary $3$-read $3$-deletion correcting code with redundancy $5\log(n)+O(1)$, and a non-binary $3$-read $3$-deletion correcting code with redundancy $7\log(n)+O(1)$.

A comparison between these sporadic constructions and the general constructions of the previous section reveals that the condition $\ell\ge 2t$ is not necessary for achieving redundancy $t\log n+O(1)$. It also shows that the joint condition $\ell=2t-1$ with $t\ge 3$ is not necessary for achieving redundancy $(2t-1)\log n+O(1)$.

\subsection{The Case $(\ell,t)=(2,2)$}\label{subsec_l2t2}
This subsection focuses on the case $(\ell,t)=(2,2)$. Chee and Vu previously constructed a binary $2$-read $2$-deletion correcting code with redundancy $4\log n+O(1)$ \cite[Theorem 9]{YeowVan2020ISIT}. We first reduce this redundancy to $2\log n+O(1)$, and then construct a $q$-ary $2$-read $2$-deletion correcting code with redundancy $4\log n+O(1)$ for any $q>2$.

Let $\bmx\in\Sigma_q^n$ and $\cR_{2}\parenv{\tilde{\bmx}}=\parenv{\bms_0,\bms_1,\ldots,\bms_{n+2}}$, where $\bms_i=\tilde{\bmx}_{[i-1,i]}$ for all $0\le i\le n+2$. Suppose that $\cR^\prime$ is obtained from $\cR_{2}\parenv{\tilde{\bmx}}$ by deleting two $2$-mers $\bms_i$ and $\bms_j$, where $1\le i<j\le n+1$ (recall that $\bms_0$ and $\bms_{n+2}$ are fixed by assumption). Denote $\cR^\prime=\parenv{\bms^\prime_0,\bms^\prime_1,\ldots,\bms^\prime_n}$.

First, assume that $j\ge i+2$. Equivalently, the $\ell$-mer $\bms_{i+1}$ is not deleted. Then the only possibly inconsistent pairs of adjacent $2$-mers in $\cR^\prime$ are $\parenv{\bms^\prime_{i-1},\bms^\prime_i}$ and $\parenv{\bms^\prime_{j-2},\bms^\prime_{j-1}}$. It follows from \Cref{thm_multiburst} that one can insert at most two $2$-mers into $\cR^\prime$ to obtain a consistent vector $\cR^{\prime\prime}$ of $2$-mers. Let $\bmy=\cR_{2}^{-1}\parenv{\cR^{\prime\prime}}$.
By \Cref{thm_multiburst}, the relationship between $\tilde{\bmx}$ and $\bmy$ has one of the following three forms:
\begin{itemize}
    \item If there are two inconsistent pairs of $2$-mers in $\cR^\prime$, then $\bmy=\tilde{\bmx}$.
    \item If there is exactly one inconsistent pair of $2$-mers in $\cR^\prime$, then $\bmy$ is obtained from $\tilde{\bmx}$ by deleting one symbol from a run\footnote{A \emph{run} in a sequence $\bmz$ is a maximal $1$-periodic substring of $\bmz$.} of length at least $2$.
    \item If $\cR^\prime$ is consistent, then $\bmy$ is obtained from $\tilde{\bmx}$ by deleting one symbol from the run containing $\tilde{x}_{i-1}$ and $\tilde{x}_{i}$, and another symbol from the run containing $\tilde{x}_{j-1}$ and $\tilde{x}_{j}$. Notice that these two runs are not necessarily distinct.
\end{itemize}

Next, consider the case $j=i+1$, where the two $2$-mers $\bms_i$ and $\bms_{i+1}$ are deleted. Then $\parenv{\bms^\prime_{i-1},\bms^\prime_i}$ is the only possible inconsistent pair of adjacent $2$-mers in $\cR^\prime$. There are two cases for the relationship between $\tilde{\bmx}$ and $\bmy$:
\begin{itemize}
    \item If $\cR^\prime$ is consistent, then $\tilde{x}_{i-1}=\tilde{x}_{i+1}$ and $\bmy$ is obtained from $\tilde{\bmx}$ by deleting $\tilde{x}_{i}$ and $\tilde{x}_{i+1}$. In particular, if $\tilde{x}_{i}=\tilde{x}_{i+1}$, then $\bmy$ is obtained from $\tilde{\bmx}$ by deleting two symbols from a run of length at least $3$.
    \item If $\cR^\prime$ is inconsistent, then $\tilde{x}_{i-1}\ne\tilde{x}_{i+1}$. The sequence $\bmy$ is obtained from $\tilde{\bmx}$ by deleting $\tilde{x}_i$. When $q=2$, it must be that $\tilde{x}_i\in\{\tilde{x}_{i-1},\tilde{x}_{i+1}\}$. Consequently, the sequence $\bmy$ is obtained by deleting one symbol from a run of length at least $2$ in $\tilde{\bmx}$.
\end{itemize}

From this analysis we obtain the following lemma.
\begin{lemma}\label{lem_l2t2}
    Let $\bmx$, $\cR^\prime$ and $\bmy$ be as above. There are at most two inconsistent pairs of $2$-mers in $\cR^\prime$.
    \begin{enumerate}[$(i)$]
        \item If $\cR^\prime$ is consistent, then $\bmy$ is obtained from $\tilde{\bmx}$ either by deleting one symbol from each of two runs of length at least $2$, or by deleting two symbols from a run of length at least $3$, or by a $2$-burst-deletion in an alternating substring\footnote{An alternating substring is a periodic substring of length at least $2$ whose minimum period is $2$.} of length at least $3$.
        \item If there is exactly one inconsistent pair of $2$-mers in $\cR^\prime$, then $\bmy$ is obtained from $\tilde{\bmx}$ either by deleting one symbol from a run of length at least $2$, or by deleting $\tilde{x}_i$ for some $1\le i\le n$ with $\tilde{x}_{i-1}$, $\tilde{x}_{i}$ and $\tilde{x}_{i+1}$ being mutually distinct (possible only when $q>2$).
        \item If there are two inconsistent pairs of $2$-mers in $\cR^\prime$, then $\bmy=\tilde{\bmx}$.
    \end{enumerate}
\end{lemma}

This lemma forms the basis for the code constructions presented later. When $\cR^\prime$ is consistent, \Cref{lem_l2t2} says that $\bmy$ is obtained from $\tilde{\bmx}$ either by deleting symbols from runs, or by a $2$-burst-deletion in an alternating substring of length at least $3$. Let $r\parenv{\tilde{\bmx}}$ be the number of runs in $\tilde{\bmx}$. By assumption, the last two symbols of $\tilde{\bmx}$ are not deleted. Then it is easy to verify that $r\parenv{\tilde{\bmx}}=r\parenv{\bmy}$ in the former case, and $r\parenv{\tilde{\bmx}}=r\parenv{\bmy}+2$ in the latter case. Therefore, knowing $r\parenv{\tilde{\bmx}}\pmod{3}$ is sufficient to distinguish between these two cases.

Let $\tilde{\bmx}_{[i_{j,1},i_{j,2}]}$ be the $j$-th run in $\tilde{\bmx}$, where $i_{j,2}<i_{j+1,1}$ for $1\le j< r\parenv{\tilde{\bmx}}$. Clearly, these runs form a partition of $\tilde{\bmx}$. For all $-1\le i\le n+2$, set $f_{\textup{run}}(\tilde{\bmx},i)=j$, if $i\in[i_{j,1},i_{j,2}]$.
For $k\ge1$, define $F_{\textup{run},k}\parenv{\tilde{\bmx}}\triangleq\sum_{i=-1}^{n+2}f_{\textup{run}}(\tilde{\bmx},i)^k$.

We first consider the binary case. To handle the situation in which $\bmy$ is obtained from $\tilde{\bmx}$ by a $2$-burst-deletion in an alternating substring of length at least $3$, we need the following lemma. This lemma is essentially part of the main result in \cite{VL1967}, although it is stated there in a different form. Since that paper is not readily accessible online, we include the proof for completeness. The proof also gives an explicit decoding procedure for recovering $\tilde{\bmx}$ from $\bmy$.
\begin{lemma}(\cfcitep{VL1967})\label{lem_binary2burst}
  Let $m\ge5$ and $N\ge 2m$. For $0\le a<2m$, define the code $\cC=\mathset{\bmz\in\Sigma_2^m:F_{\textup{run},1}\parenv{\bmz}\equiv a\pmod{N}}$. Suppose that $\bmy$ is obtained from some $\bmz\in\cC$ by deleting $z_{i}$ and $z_{i+1}$ for some $i$ satisfying $z_{i+1}=z_{i-1}$ and $z_{i}\ne z_{i-1}$, where $1<i<m-1$. Then $\bmz$ can be efficiently decoded from $\bmy$.
\end{lemma}
\begin{IEEEproof}
   The length of $\bmy$ is $m-2$. For $k\in[m-3]$, let $\bmz^{(k)}=\bmy_{[1,k]}(1-y_{k})y_k\bmy_{[k+1,m-2]}$. That is, sequence $\bmz^{(k)}$ is obtained from $\bmy$ by inserting the two bits $(1-y_k)$ and $y_k$ (in this order) between the $k$-th and the $(k+1)$-th positions. It is clear that
   \begin{equation*}
       f_{\textup{run}}\parenv{\bmz^{(k)},j}=
       \begin{cases}
           f_{\textup{run}}\parenv{\bmy,j},\mbox{ if }j\le k,\\
           f_{\textup{run}}\parenv{\bmy,k}+1,\mbox{ if }j= k+1,\\
           f_{\textup{run}}\parenv{\bmy,k}+2,\mbox{ if }j= k+2,\\
           f_{\textup{run}}\parenv{\bmy,j-2}+2,\mbox{ if }j\ge k+3.
       \end{cases}
   \end{equation*}
   Let $\Delta^{(k)}=F_{\textup{run},1}\parenv{\bmz^{(k)}}-F_{\textup{run},1}\parenv{\bmy}$. Then $\Delta^{(k)}=2f_{\textup{run}}\parenv{\bmy,k}+3+2(m-2-k)$.

   For all $k<k^\prime$, we have $\Delta^{(k)}-\Delta^{(k^\prime)}=-2\parenv{f_{\textup{run}}\parenv{\bmy,k^\prime}-f_{\textup{run}}\parenv{\bmy,k}}+2\parenv{k^\prime-k}\ge0$, because $f_{\textup{run}}\parenv{\bmy,k^\prime}-f_{\textup{run}}\parenv{\bmy,k}\le k^\prime-k$. Therefore, we conclude that $0\le \Delta^{(k)}\le\Delta^{(1)}=2m-1<N$. Moreover, equality $\Delta^{(k)}=\Delta^{(k^\prime)}$ holds if and only if $f_{\textup{run}}\parenv{\bmy,k^\prime}-f_{\textup{run}}\parenv{\bmy,k}=k^\prime-k$, which is equivalent to $\bmy_{s+1}=1-\bmy_{s}$ for all $k\le s<k^\prime$. Consequently, it holds that $\Delta^{(k)}=\Delta^{(k^\prime)}$ if and only if $\bmz^{(k)}=\bmz^{(k^\prime)}$.

   Let $\delta$ be the unique integer in $[0,N-1]$ such that $\delta\equiv a-F_{\textup{run},1}\parenv{\bmy}\pmod{N}$. Choose the smallest $k\in[m-3]$ such that $\Delta^{(k)}=\delta$. By the relationship between $\bmy$ and $\bmz$, such an index $k$ must exist. The preceding argument then guarantees that $\bmz^{(k)}=\bmz$, completing the proof.
\end{IEEEproof}

We now give the binary $2$-read $2$-deletion correcting code.
\begin{theorem}\label{thm_binaryl2t2}
   For $n\ge2$, let $p>2n+8$ be the smallest prime. For given $0\le a_1,a_2<p$ and $0\le b<3$, set $\bma=\parenv{a_1,a_2}$ and define the code
     \begin{equation*}
        \cC(n;\bma,b)=\mathset{\bmx\in\Sigma_2^n:
         \begin{array}{c}
        F_{\textup{run},k}\parenv{\tilde{\bmx}}\equiv a_k\pmod{p},\forall k=1,2\\
        r(\tilde{\bmx})\equiv b\pmod{3}
        \end{array}
        }.
    \end{equation*}
    This code can correct two deletions in $\cR_{2}(\tilde{\bmx})$. There exist some $\bma$ and $b$ such that $\rho\parenv{\cC(n;\bma,b)}\le 2\log\parenv{p}+\log 3=2\log n+O(1)$.
\end{theorem}
\begin{IEEEproof}
Suppose that $\tilde{\bmx}=\bmu\bmx\bmv$ is transmitted, and that $\cR^\prime$ is the received vector of $2$-mers, where $\bmx\in\cC(n;\bma,b)$ and $\cR^\prime$ is a subvector of $\cR_{2}(\tilde{\bmx})$ of length $n+1$. Let $\bmy$ be the sequence derived from $\cR^\prime$ as described at the beginning of this subsection. We then show how to decode $\tilde{\bmx}$ from $\bmy$.

By \Cref{lem_l2t2}, if there are two inconsistent pairs of $2$-mers in $\cR^\prime$, then $\tilde{\bmx}=\bmy$. Hence no further decoding is needed.

If there is exactly one inconsistent pair of $2$-mers in $\cR^\prime$, then $\bmy$ is obtained from $\tilde{\bmx}$ by deleting one symbol from a run of length at least $2$. In this case, $r\parenv{\tilde{\bmx}}=r(\bmy)\le\abs{\bmy}= n+3$. Suppose that the deletion occurred in the $i$-th run. Then $1\le i=F_{\textup{run},1}\parenv{\tilde{\bmx}}-F_{\textup{run},1}\parenv{\bmy}\le n+3<p$, implying that $i$ is the unique integer in $[p-1]$ such that $i\equiv a_1-F_{\textup{run},1}\parenv{\bmy}\pmod{p}$. Once $i$ is known, $\tilde{\bmx}$ can be recovered by extending the $i$-th run of $\bmy$ by one symbol.

It remains to consider the case where $\cR^\prime$ is consistent. If $r(\bmy)\equiv b\pmod{3}$, then by the preceding discussion, sequence $\bmy$ is obtained from $\tilde{\bmx}$ either by deleting one symbol from each of two certain runs of length at least $2$, or by deleting two symbols from a run of length at least $3$. In both cases, we have $r\parenv{\tilde{\bmx}}=r(\bmy)\le\abs{\bmy}= n+2$. Assume that the two deleted symbols lie in the $i$-th and $j$-th runs, where $i\le j$. Then
\begin{equation}\label{eq_l2t2}
\left\{
\begin{array}{r}
    i+j=F_{\textup{run},1}\parenv{\tilde{\bmx}}-F_{\textup{run},1}\parenv{\bmy} \equiv a_1-F_{\textup{run},1}\parenv{\bmy}\pmod{p},\\
    i^2+j^2=F_{\textup{run},2}\parenv{\tilde{\bmx}}-F_{\textup{run},2}\parenv{\bmy} \equiv a_2-F_{\textup{run},2}\parenv{\bmy}\pmod{p}.
\end{array}
\right.
\end{equation}
By \Cref{lem_pre1}, the values of $i$ and $j$ can be uniquely determined by \eqref{eq_l2t2}. Therefore, $\tilde{\bmx}$ can be recovered by extending the corresponding run(s).

Finally, suppose that $r(\bmy)\equiv b-2\pmod 3$. Then according to previous discussion, sequence $\bmy$ is obtained from $\tilde{\bmx}$ by a $2$-burst-deletion in an alternating substring of length at least $3$. In this case, $\tilde{\bmx}$ can be recovered using the decoding procedure described in the proof of \Cref{lem_binary2burst}. Now the proof is completed.
\end{IEEEproof}

We next turn to the construction of non-binary codes. Since \Cref{lem_binary2burst} does not extend directly to the case $q>2$, a different method is required to handle the case where $\bmy$ is obtained from $\tilde{\bmx}$ by a $2$-burst-deletion in an alternating substring of length at least $3$.

For integers $m>T\ge2$ and a sequence $\bmz\in\Sigma_q^m$, let $\hat{\bmz}$ be the sequence obtained from $\bmz$ by padding $m-\floorenv{m/T}T$ symbols $0$. Then we have $T\mid\abs{\hat{\bmz}}$. Partition $\hat{\bmz}$ into non-overlapping segments of length $T$ as $\hat{\bmz}=\hat{\bmz}_{[1,T]}\hat{\bmz}_{[T+1,2T]}\cdots$. In other words, the $i$-th segment is $\hat{\bmz}_{[(i-1)T+1,iT]}$. Define a mapping $p_{T}:\Sigma_q^m\rightarrow\Sigma_{q^{T}}^{\ceilenv{m/T}}$ as follows. For $\bmz\in\Sigma_q^m$, the symbol in the $i$-th position of $p_{T}(\bmz)$ is $\sum_{j=1}^{T}\hat{z}_{(i-1)T+j}q^{j-1}$.

\begin{lemma}\label{lem_bursttosingle}
    Suppose that $T\ge2$ and $m\ge 2T-1$. Let $\bmy$ be obtained from $\bmz$ by a $T$-burst-deletion in a $T$-periodic substring of length at least $2T-1$. Then $p_T(\bmy)$ is obtained from $p_T(\bmz)$ by a single deletion.
\end{lemma}
\begin{IEEEproof}
    Suppose that the $T$-periodic substring is $\bmz_{[s,s+L-1]}$, where $L\ge2T-1$ and $1\le s\le m-L+1$. Assume $s=(i-1)T+j$ for some $i\ge1$ and $1\le j\le T$. Since $\bmz_{[s,s+L-1]}$ is $T$-periodic and $\bmy$ is obtained by a $T$-burst-deletion in $\bmz_{[s,s+L-1]}$, it follows that $\bmy$ can be obtained by deleting arbitrary $T$ consecutive symbols in $\bmz_{[s,s+L-1]}$.

    If $j=1$, then $\bmz_{[s,s+L-1]}$ contains the $i$-th segment since $L\ge2T-1\ge T$. In this case, $p_T(\bmy)$ is obtained from $p_T(\bmz)$ by deleting the $i$-th symbol. Otherwise, $\bmz_{[s,s+L-1]}$ contains the $(i+1)$-th segment because $L-(T+1-j)\ge T$. In this case, $p_T(\bmy)$ is obtained from $p_T(\bmz)$ by deleting the $(i+1)$-th symbol.
\end{IEEEproof}

Inspired by \Cref{lem_bursttosingle}, we need a non-binary code that corrects a single deletion. Let $m\ge2$ and $q\ge2$. For a sequence $\bmz\in\Sigma_q^m$, define a $q$-ary sequence $\psi(\bmz)\in\Sigma_q^m$ as
\begin{equation*}
    \psi(\bmz)_i=
    \begin{cases}
        z_i-z_{i+1}\pmod{q},\mbox{ if }i<m,\\
        z_m,\mbox{ if }i=m.
    \end{cases}
\end{equation*}
The Varshamov-Tenengolts syndrome of $\bmz$ is defined to be $\vt\parenv{\bmz}=\sum_{i=1}^{m}iz_i$.

\begin{lemma}(\cfcite{Tuan2024IT}{Theorem 3})\label{lem_qary1del}
    Suppose that $\bmy$ is obtained from $\bmz\in\Sigma_q^m$ by deleting one symbol. Then one can decode $\bmz$ from $\bmy$ efficiently when given $\vt\parenv{\psi(\bmz)}\pmod{mq}$.
\end{lemma}

The following corollary follows immediately from \Cref{lem_bursttosingle,lem_qary1del}.
\begin{corollary}\label{cor_TburstTperiodic}
    Suppose that $q,T\ge2$ and $m\ge 2T-1$. For integer $0\le a<q^T\ceilenv{m/T}$, define the code $$
    \cC=\mathset{\bmz\in\Sigma_q^m:\vt\parenv{\psi\parenv{p_T(\bmz)}}\equiv a\pmod{q^T\ceilenv{m/T}}}.
    $$
    Let $\bmy$ be obtained from $\bmz\in\cC$ by a $T$-burst-deletion in a $T$-periodic substring of length at least $2T-1$. Then one can efficiently recover $\bmz$. There exists some $a$ such that $\rho(\cC)\le \log\parenv{q^T\ceilenv{m/T}}$.
\end{corollary}

\begin{remark}
    For $T=2$, Wang \emph{et al} in \cite[Theorem 7]{Shuche2024IT} constructed a code with redundancy $\log\parenv{2mq^2}$. In comparison, our construction has redundancy $\log\parenv{q^2\ceilenv{m/2}}$. In addition to improving the redundancy, our construction works for general $T\ge2$.
\end{remark}

\begin{remark}
    For any $q\ge2$, Sun \emph{et al} in \cite[Theorem 9]{Yubo2025IT} constructed a $q$-ary code of length $m$ that can correct a $T$-burst-deletion. Therefore, their code can also correct a $T$-burst-deletion in a $T$-periodic substring. However, their code has redundancy $\log n+\log\parenv{\frac{m}{T}(18)^{\floorenv{\frac{(T+1)^2}{2}}+1}4^{q^T-1}q^{4T}}$, which is much greater than the redundancy of the code in \Cref{cor_TburstTperiodic}.
\end{remark}

We are now ready to present the non-binary construction.
\begin{theorem}\label{thm_qaryl2t2}
    For $q>2$ and $n\ge2$, let $p\ge n+3$ be the smallest prime. For given $0\le a_1,a_2<p$, $0\le b_1<q^2\ceilenv{\frac{n+4}{2}}$, $0\le b_2<(n+4)q$ and $0\le c<3$, set $\bma=\parenv{a_1,a_2}$ and $\bmb=\parenv{b_1,b_2}$. Define the code
     \begin{equation*}
        \cC(n;\bma,\bmb,c)=\mathset{\bmx\in\Sigma_q^n:
         \begin{array}{c}
        F_{\textup{run},k}\parenv{\tilde{\bmx}}\equiv a_k\pmod{p},\forall k=1,2\\
        \vt\parenv{\psi\parenv{p_2(\tilde{\bmx})}}\equiv b_1\pmod{q^2\ceilenv{\frac{n+4}{2}}}\\
        \vt\parenv{\psi(\tilde{\bmx})}\equiv b_2\pmod{(n+4)q}\\
        r(\tilde{\bmx})\equiv c\pmod{3}
        \end{array}
        }.
    \end{equation*}
    This code can correct two deletions in $\cR_{2}(\tilde{\bmx})$. There exist some $\bma$, $\bmb$ and $c$ such that $\rho\parenv{\cC(n;\bma,\bmb,c)}\le 2\log\parenv{p}+\log\parenv{q^2\ceilenv{\frac{n+4}{2}}}+\log((n+4)q)+\log 3=4\log n+O(1)$.
\end{theorem}
\begin{IEEEproof}
  Suppose that $\tilde{\bmx}=\bmu\bmx\bmv$ is transmitted, and that $\cR^\prime$ is the received vector of $2$-mers, where $\bmx\in\cC(n;\bma,\bmb,c,\bmd)$ and $\cR^\prime$ is a subvector of $\cR_{2}(\tilde{\bmx})$ of length $n+1$. Let $\bmy$ be the sequence derived from $\cR^\prime$ as described at the beginning of this subsection. We then show how to decode $\tilde{\bmx}$ from $\bmy$. We only have to consider the case when $\cR^\prime$ is consistent or contains exactly one inconsistent pair of $2$-mers.

  If $\cR^\prime$ contains exactly one inconsistent pair of $2$-mers, then $\bmy$ is obtained from $\tilde{\bmx}$ by deleting one symbol. Thus, \Cref{lem_qary1del}, together with the condition $\operatorname{VT}(\psi(\tilde{\bmx}))\equiv b_2\pmod{(n+4)q}$, allows us to recover $\tilde{\bmx}$.

  Now suppose that $\cR^\prime$ is consistent. If $r(\bmy)\equiv c\pmod{3}$, we can recover $\tilde{\bmx}$ as in the proof of \Cref{thm_binaryl2t2}. If $r(\bmy)\equiv c-2\pmod{3}$, the sequence $\bmy$ is obtained from $\tilde{\bmx}$ by a $2$-burst-deletion in an alternating substring of length at least $3$. This case can be handled by \Cref{cor_TburstTperiodic} together with the condition $\vt\parenv{\psi\parenv{p_2(\tilde{\bmx})}}\equiv b_1\pmod{q^2\ceilenv{\frac{n+4}{2}}}$.
\end{IEEEproof}

\subsection{The Case $(\ell,t)=(3,2)$}\label{subsec_l3t2}
For $\ell\ge 1$ and $q\ge 2$, let $\mathscr{M}_{q,\ell}$ denote the set of all finite vectors of $\ell$-mers over $\Sigma_q$. Given $\cR_1=\parenv{\bmz_1,\ldots,\bmz_m}\in\mathscr{M}_{q,\ell+1}$, we decompose each $(\ell+1)$-mer $\bmz_i$ into its length-$\ell$ prefix  $\bmz_{i,1}$ and length-$\ell$ suffix $\bmz_{i,2}$. This gives the vector of $\ell$-mers $\parenv{\bmz_{1,1},\bmz_{1,2},\ldots,\bmz_{m,1},\bmz_{m,2}}$. For each $1\le i<m$, whenever $\bmz_{i,2}=\bmz_{i+1,1}$, we delete $\bmz_{i+1,1}$. The resulting vector of $\ell$-mers is denoted by $\cR_2$. We write $\cT_{\ell}$ for the mapping that sends $\cR_1$ to $\cR_2$.

\begin{lemma}\label{lem_iteratemer}
   For any $\bmx\in\Sigma_q^n$ with $n\ge\ell+1$, we have $\cT_{\ell}\parenv{\cR_{\ell+1}(\bmx)}=\cR_{\ell}\parenv{\bmx}$. Moreover, if $\cR^\prime$ is obtained from $\cR_{\ell+1}(\bmx)$ by deleting $t$ $(\ell+1)$-mers, then $\cT_{\ell}\parenv{\cR^\prime}$ is obtained from $\cR_{\ell}\parenv{\bmx}$ by deleting \emph{at most} $t$ $\ell$-mers.
\end{lemma}
\begin{IEEEproof}
    The identity $\cT_{\ell}\parenv{\cR_{\ell+1}(\bmx)}=\cR_{\ell}\parenv{\bmx}$ follows directly from the definitions of $\cR_{\ell+1}(\bmx)$, $\cR_{\ell}(\bmx)$ and $\cT_{\ell}$.

    It remains to prove the second claim. We first consider the case where $\cR^\prime$ is obtained from $\cR_{\ell+1}(\bmx)$ by deleting $t$ consecutive $(\ell+1)$-mers, namely $\bmx_{[i+1,i+\ell+1]},\bmx_{[i+2,i+\ell+2]},\ldots,\bmx_{[i+t,i+\ell+t]}$, for some $0\le i\le n-\ell-t$. If $i=0$, then by the definition of $\cT_{\ell}$, we have $\cT_{\ell}\parenv{\cR^\prime}=\parenv{\bmx_{[t+1,t+\ell]},\bmx_{[t+2,t+\ell+1]},\ldots,\bmx_{[n-\ell+1,n]}}$. Thus,  $\cT_{\ell}\parenv{\cR^\prime}$ is obtained from $\cR_{\ell}(\bmx)$ by deleting the first $t$ $\ell$-mers. The case $i=n-\ell-t$ is analogous, in which $\cT_{\ell}\parenv{\cR^\prime}$ is obtained from $\cR_{\ell}(\bmx)$ by deleting the last $t$ $\ell$-mers.

    Now suppose that $0<i<n-\ell-t$. If $\bmx_{[i+1,i+\ell]}\ne\bmx_{[i+t+1,i+t+\ell]}$, then the definition of $\cT_{\ell}$ gives
    $$
    \cT_{\ell}\parenv{\cR^\prime}=\parenv{\bmx_{[1,\ell]},\ldots,\bmx_{[i+1,i+\ell]},\bmx_{[i+t+1,i+t+\ell]},\ldots,\bmx_{[n-\ell+1,n]}},
    $$
    which is obtained from $\cR_{\ell}(\bmx)$ by deleting $t-1$ $\ell$-mers. On the other hand, if $\bmx_{[i+1,i+\ell]}=\bmx_{[i+t+1,i+t+\ell]}$, then
    $$
    \cT_{\ell}\parenv{\cR^\prime}=\parenv{\bmx_{[1,\ell]},\ldots,\bmx_{[i+1,i+\ell]},\bmx_{[i+t+2,i+t+\ell+1]},\ldots,\bmx_{[n-\ell+1,n]}}.
    $$
    In this case, $\cT_{\ell}\parenv{\cR^\prime}$ is obtained from $\cR_{\ell}(\bmx)$ by deleting $t$ $\ell$-mers.

    Finally, arbitrary $t$ deletions in $\cR_{\ell+1}(\bmx)$ can be decomposed into disjoint bursts of consecutive deletions. Applying the preceding argument to each burst shows that $\cT_\ell(\cR^{\prime})$ is obtained from $\cR_\ell(\bmx)$ by deleting at most $t$ $\ell$-mers. This completes the proof.
\end{IEEEproof}

Recall that, under the $(\ell+1)$-read channel model, the actual transmitted sequence associated with an input sequence $\bmx\in\Sigma_q^n$ is $\bmu\bmx\bmv$, where $\bmu=u_1\cdots u_{\ell+1}$ and $\bmv=v_1\cdots v_{\ell+1}$ are two fixed sequences in $\Sigma_q^{\ell+1}$. The next lemma follows immediately from \Cref{lem_iteratemer}.
\begin{lemma}\label{lem_iteratecode}
    Let $\cC\subseteq\Sigma_q^n$ be a code. Suppose that, for every $\bmx\in\cC$, the code $\cC$ can correct any $t$ deletions in the vector $\cR_{\ell}\parenv{u_2\cdots u_{\ell+1}\bmx v_1\cdots v_{\ell}}$. Then $\cC$ is an $(\ell+1)$-read $t$-deletion correcting code.
\end{lemma}

Applying this lemma together with \Cref{thm_binaryl2t2}, we obtain the following result.
\begin{corollary}\label{cor_bianryl3t2}
    There is a binary $3$-read $2$-deletion correcting code of length $n$ whose redundancy is at most $2\log(n)+O(1)$.
\end{corollary}

Similarly, combining \Cref{lem_iteratecode} with \Cref{thm_qaryl2t2} yields a $q$-ary $3$-read $2$-deletion correcting code of length $n$ with redundancy at most $4\log(n)+O(1)$ for every $q>2$. In what follows, we show that this redundancy can be reduced to $3\log(n)+O(1)$. The key point is that when $\ell\ge3$, the second situation in \Cref{lem_l2t2} (ii) never happens.

Let $\bmx\in\Sigma_q^n$ and write $\cR_{3}\parenv{\tilde{\bmx}}=\parenv{\bms_0,\bms_1,\ldots,\bms_{n+3}}$, where $\bms_i=\tilde{\bmx}_{[i-2,i]}$ for all $0\le i\le n+3$. Suppose that $\cR^\prime$ is obtained from $\cR_{3}\parenv{\tilde{\bmx}}$ by deleting two $3$-mers $\bms_i$ and $\bms_j$, where $1\le i<j\le n+2$. Denote the received vector by $\cR^\prime=\parenv{\bms^\prime_0,\bms^\prime_1,\ldots,\bms^\prime_{n+1}}$.

By \Cref{thm_multiburst}, one can insert at most two $3$-mers into $\cR^\prime$ to get a consistent vector $\cR^{\prime\prime}$ of $3$-mers. Let $\bmy=\cR_{2}^{-1}\parenv{\cR^{\prime\prime}}$.
As in the case $(\ell,t)=(2,2)$, when $j\ge i+2$, the relationship between $\tilde{\bmx}$ and $\bmy$ falls into the following cases.
\begin{itemize}
    \item If $\cR^\prime$ contains two inconsistent pairs of $3$-mers, then $\bmy=\tilde{\bmx}$.
    \item If $\cR^\prime$ contains exactly one inconsistent pair of $3$-mers, then $\bmy$ is obtained from $\tilde{\bmx}$ by deleting one symbol from a run of length at least $3$.
    \item If $\cR^\prime$ is consistent, then $\bmy$ is obtained from $\tilde{\bmx}$ by deleting one symbol from each of two runs of length at least $3$.
\end{itemize}

Now suppose that $j=i+1$. In this case, the vector $\cR^\prime$ contains at most one pair of inconsistent $3$-mers. There are three possibilities.
\begin{itemize}
    \item If $\cR^\prime$ is consistent, then $\bmy$ is obtained from $\tilde{\bmx}$ by a $2$-burst-deletion in a substring of length at least $4$. This substring is either a run or an alternating substring.
    \item If $\cR^\prime$ is inconsistent and $\abs{\bmy}=n+5$ ($=\abs{\bmx}-1$), then $\bmy$ is obtained by deleting one symbol from a run of length at least $2$.
    \item If $\cR^\prime$ is inconsistent and $\abs{\bmy}=n+6$, then $\bmy=\tilde{\bmx}$.
\end{itemize}

The preceding discussion yields the following lemma.
\begin{lemma}\label{lem_l3t2}
    Let $\tilde{\bmx}$ and $\bmy$ be as above. Since $|\tilde{\bmx}|=n+6$, the following statements hold.
    \begin{enumerate}[$(i)$]
        \item If $\abs{\bmy}=n+4$, then $\bmy$ is obtained from $\tilde{\bmx}$ either by deleting one symbol from each of two runs of length at least $3$, or by deleting two symbols from a run of length at least $4$, or by a $2$-burst-deletion in an alternating substring of length at least $4$.
        \item If $\abs{\bmy}=n+5$, then $\bmy$ is obtained from $\tilde{\bmx}$ by deleting one symbol from a run of length at least $2$.
        \item If $\abs{\bmy}=n+6$, then $\bmy=\tilde{\bmx}$.
    \end{enumerate}
\end{lemma}

This lemma already implies a $q$-ary $3$-read $2$-deletion correcting code via a modification of the construction of \Cref{thm_qaryl2t2}: one may remove the constraint $\vt\parenv{\psi(\tilde{\bmx})}\equiv b_2\pmod{(n+4)q}$, since the second situation in \Cref{lem_l2t2} (ii) cannot occur.

We next give an alternative construction for $q>2$. This construction is simpler and achieves slightly less redundancy. Recall that, in the $2$-read case, one possible situation is that $\bmy$ is obtained from $\tilde{\bmx}$ by deleting two consecutive symbols from an alternating substring of length at least $3$. There, \Cref{cor_TburstTperiodic} was used to recover $\tilde{\bmx}$. In the present $3$-read setting, however, \Cref{lem_l3t2} (i)  shows that the corresponding alternating substring has length at least $4$. Hence, after two consecutive deletions, at least two symbols of this alternating substring remain. Therefore, it suffices to identify which maximal alternating substring suffered the deletions.

For a sequence $\bmz$, let $m(\bmz)$ denote the number of maximal alternating substrings in $\bmz$. If $m(\tilde{\bmx})=0$, set $f_{\textup{alt}}(\tilde{\bmx},i)=0$ for all $-2\le i\le n+3$. If $m(\tilde{\bmx})\ge 1$, let $\tilde{\bmx}_{[i_{j,1},i_{j,2}]}$ be the $j$-th maximal alternating substring of $\tilde{\bmx}$, indexed so that $i_{j,1}<i_{j+1,1}$ and $i_{j,2}<i_{j+1,2}$ for every $1\le j< m(\tilde{\bmx})$. For all $-\ell+1\le i\le n+\ell$, set
\begin{equation*}
    f_{\textup{alt}}(\tilde{\bmx},i)=
    \begin{cases}
        j,\mbox{ if }i\in[i_{j,1},i_{j,2}]\setminus[i_{j-1,1},i_{j-1,2}],\\
        0,\mbox{ else}.
    \end{cases}
\end{equation*}
Define $F_{\textup{alt}}\parenv{\tilde{\bmx}}\triangleq\sum_{i=-2}^{n+3}f_{\textup{alt}}(\tilde{\bmx},i)$.

Let $F_{\textup{run},1}\parenv{\tilde{\bmx}}$ and $F_{\textup{run},2}\parenv{\tilde{\bmx}}$ be defined as in \Cref{subsec_l2t2}.
\begin{theorem}\label{thm_qaryl3t2}
    For $q>2$ and $n\ge2$, let $p\ge n+5$ be the smallest prime. Let $N=n+3$ if $n$ is even, and $N=n+4$ otherwise. For given $0\le a_1,a_2<p$, $0\le b<N$, and $0\le c<3$, set $\bma=\parenv{a_1,a_2}$. Define the code
     \begin{equation*}
        \cC(n;\bma,b,c)=\mathset{\bmx\in\Sigma_q^n:
         \begin{array}{c}
        F_{\textup{run},k}\parenv{\tilde{\bmx}}\equiv a_k\pmod{p},\forall k=1,2\\
        F_{\textup{alt}}\parenv{\tilde{\bmx}}\equiv b\pmod{N}\\
        r(\tilde{\bmx})\equiv c\pmod{3}
        \end{array}
        }.
    \end{equation*}
    This code can correct two deletions in $\cR_{3}(\tilde{\bmx})$. Moreover, there are some $\bma$, $b$ and $c$ such that $\rho\parenv{\cC(n;\bma,b,c)}\le 2\log\parenv{p}+\log(N)+\log 3=3\log n+O(1)$.
\end{theorem}
\begin{IEEEproof}
    Suppose that $\tilde{\bmx}=\bmu\bmx\bmv$ is transmitted. Let $\cR^\prime$ be the received vector of $3$-mers, where $\bmx\in\cC(n;\bma,b,c)$ and $\cR^\prime$ is a subvector of $\cR_{3}(\tilde{\bmx})$ of length $n+2$. Let $\bmy$ be the sequence derived from $\cR^\prime$ as described above. We show how to recover $\tilde{\bmx}$ from $\bmy$.

    By \Cref{lem_l3t2}, if $|\bmy|=n+6$, then $\bmy=\tilde{\bmx}$. Therefore, it remains to consider the cases $|\bmy|=n+5$ and $|\bmy|=n+4$.

    First, suppose that $\abs{\bmy}=n+5$. By \Cref{lem_l3t2}, the sequence $\bmy$ is obtained from $\tilde{\bmx}$ by deleting one symbol from a run of length at least $2$. Hence $r(\tilde{\bmx})=r(\bmy)\le n+5$. As in the proof of \Cref{thm_binaryl2t2}, the condition $F_{\textup{run},1}(\tilde{\bmx})\equiv a_1\pmod{p}$ uniquely identifies the run in which the deletion occurred. Once this run is known, $\tilde{\bmx}$ is recovered by extending the corresponding run of $\bmy$ by one symbol.

    Now suppose that $\abs{\bmy}=n+4$. If $r(\bmy)\equiv c\pmod{3}$, then $\bmy$ is obtained from $\tilde{\bmx}$ by one deletion from each of two runs of length at least $3$, or by deleting two symbols from a run of length at least $4$. In this case, $\tilde{\bmx}$ can be recovered as in the proof of \Cref{thm_binaryl2t2}, using the two constraints $F_{\textup{run},1}\parenv{\tilde{\bmx}}\equiv a_1\pmod{p}$ and $F_{\textup{run},2}\parenv{\tilde{\bmx}}\equiv a_2\pmod{p}$.

    It remains to handle the case $r(\bmy)\equiv c-2\pmod{3}$. By \Cref{lem_l3t2}, the sequence $\bmy$ is obtained from $\tilde{\bmx}$ by deleting two consecutive symbols from an alternating substring of length at least $4$. Since at least two symbols of this alternating substring remain after the deletions, the deletion does not change the number or the order of maximal alternating substrings. Hence, we have $m(\bmy)=m(\tilde{\bmx})\ge1$. Moreover, every alternating substring has length at least $2$, and any two distinct maximal alternating substrings overlap in at most one position. Therefore, $m(\bmy)\le\abs{\bmy}-1=n+3$. Suppose that the two deletions occurred in the $i$-th maximal alternating substring. Then $2i=F_{\textup{alt}}\parenv{\tilde{\bmx}}-F_{\textup{alt}}(\bmy)\equiv b-F_{\textup{alt}}(\bmy)\pmod{N}$. Since $N$ is odd, $2$ is invertible modulo $N$. Together with the bound $1\le i\le n+3\le N$, this congruence uniquely determines $i$. Once $i$ is known, the sequence $\tilde{\bmx}$ can be recovered by extending the $i$-th maximal alternating substring of $\bmy$ by two symbols. This completes the proof.
\end{IEEEproof}

\subsection{The Case $(\ell,t)=(3,3)$}
In this subsection, we aim to construct $3$-read $3$-deletion correcting codes.

Let $\bmx\in\Sigma_q^n$ and write $\cR_{3}\parenv{\tilde{\bmx}}=\parenv{\bms_0,\bms_1,\ldots,\bms_{n+3}}$, where $\bms_i=\tilde{\bmx}_{[i-2,i]}$ for all $0\le i\le n+3$. Suppose that $\cR^\prime$ is obtained from $\cR_{3}\parenv{\tilde{\bmx}}$ by deleting three $3$-mers $\bms_i$, $\bms_j$ and $\bms_k$, where $1\le i<j<k\le n+2$. Denote the received vector by $\cR^\prime=\parenv{\bms^\prime_0,\bms^\prime_1,\ldots,\bms^\prime_{n}}$. By \Cref{thm_multiburst}, one can insert at most three $3$-mers into $\cR^\prime$ to get a consistent vector $\cR^{\prime\prime}$ of $3$-mers. Let $\bmy=\cR_{3}^{-1}\parenv{\cR^{\prime\prime}}$.

We first suppose that $j\ge i+2$ and $k\ge j+2$. Then $\cR^\prime$ contains at most $3$ inconsistent pairs of adjacent $3$-mers. By \Cref{thm_multiburst}, if $\cR^\prime$ contains $s$ inconsistent pairs, then the sequence $\bmy$ has length $n+3+s$ and is obtained from $\tilde{\bmx}$ by deleting one symbol from each of $3-s$ runs of length at least $3$.

Next, suppose that $j=i+1$ and $k\ge j+2$; the case $j\ge i+2$ and $k=j+1$ is analogous. Then $\parenv{\bms_{i-1}^\prime,\bms_{i}^{\prime}}$ and $\parenv{\bms_{k-1},\bms_{k+1}}$ are the only two possibly inconsistent pairs of adjacent $3$-mers in $\cR^\prime$. The relationship between $\tilde{\bmx}$ and $\bmy$ falls into the following cases.
\begin{itemize}
    \item If both $\parenv{\bms_{i-1},\bms_{i+2}}$ and $\parenv{\bms_{k-1},\bms_{k+1}}$ are inconsistent, then either $\bmy=\tilde{\bmx}$ or $\bmy$ is obtained from $\tilde{\bmx}$ by deleting one symbol from a run of length at least $2$.
    \item If $\parenv{\bms_{i-1},\bms_{i+2}}$ is consistent and $\parenv{\bms_{k-1},\bms_{k+1}}$ is inconsistent, then $\bmy$ is obtained from $\tilde{\bmx}$ by a $2$-burst-deletion in a substring of length at least $4$. This substring is either a run or an alternating substring.
    \item If $\parenv{\bms_{i-1},\bms_{i+2}}$ is inconsistent and $\parenv{\bms_{k-1},\bms_{k+1}}$ is consistent, then $\bmy$ is obtained from $\tilde{\bmx}$ either by deleting one symbol from a run of length at least $3$, or by deleting one symbol from a run of length at least $2$ and deleting one symbol from a run of length at least $3$.
    \item If both $\parenv{\bms_{i-1},\bms_{i+2}}$ and $\parenv{\bms_{k-1},\bms_{k+1}}$ are consistent, then $\bmy$ is obtained from $\tilde{\bmx}$ by deleting two consecutive symbols from a $2$-periodic substring of length at least $4$ and deleting one symbol from a run of length at least $3$.
\end{itemize}

Finally, suppose that $j-i=k-j=1$. In other words, the deleted $3$-mers are $\bms_i$, $\bms_{i+1}$ and $\bms_{i+2}$. Then $\cR^\prime$ contains at most one inconsistent pair of adjacent $3$-mers. There are three possibilities.
\begin{itemize}
    \item If $\cR^\prime$ is consistent, then $\bmy$ is obtained from $\tilde{\bmx}$ by a $3$-burst-deletion in a $3$-periodic substring of length at least $5$.
    \item If $\cR^\prime$ is inconsistent and $\abs{\bmy}=n+4$, then $\bmy$ is obtained from $\tilde{\bmx}$ by a $2$-burst-deletion in a substring of length at least $3$. This substring is either a run or an alternating substring.
    \item If $\cR^\prime$ is inconsistent and $\abs{\bmy}=n+5$, then $\tilde{x}_{i-1}\ne\tilde{x}_{i+1}$ and $\bmy$ is obtained from $\tilde{\bmx}$ by deleting $\tilde{x}_i$.
\end{itemize}

The preceding analysis gives the following result.
\begin{lemma}\label{lem_l3t3}
    Let $\tilde{\bmx}$, $\cR^\prime$ and $\bmy$ be as above. The following statements hold.
    \begin{enumerate}[$(i)$]
        \item If $\abs{\bmy}=n+6$ ($=\abs{\tilde{\bmx}}$), then $\bmy=\tilde{\bmx}$.
        \item If $\abs{\bmy}=n+5$, then $\bmy$ is obtained from $\tilde{\bmx}$ either by deleting one symbol from a run of length at least $2$, or by deleting $\tilde{x}_{i}$ for some $1\le i\le n$ with $\tilde{x}_{i-1}$, $\tilde{x}_i$ and $\tilde{x}_{i+1}$ being mutually distinct (only possible when $q>2$).
        \item If $\abs{\bmy}=n+4$, then $\bmy$ is obtained from $\tilde{\bmx}$ either by deleting one symbol from each of two runs of length at least $2$, or by deleting two symbols from a run of length at least $3$, or by a $2$-burst-deletion in an alternating substring of length at least $3$.
        \item If $\abs{\bmy}=n+3$, then $\bmy$ is obtained from $\tilde{\bmx}$ in one of the following ways:
        \begin{itemize}
            \item by deleting one symbol from each of three runs of length at least $3$;
            \item by deleting two symbols from a run of length at least $4$ and deleting one symbol from a run of length at least $3$;
            \item by deleting two consecutive symbols from an alternating substring of length at least $4$ and deleting one symbol from a run of length at least $3$;
            \item by replacing $\tilde{\bmx}_{[i-1,i+3]}$ with $\tilde{x}_{i-1}\tilde{x}_{i}$ for some $1\le i\le n$ satisfying $\tilde{\bmx}_{[i-1,i+3]}\in\mathset{aaaaa,aabaa,abbab,abaab,abcab}$, where $a,b,c$ are mutually distinct symbols from $\Sigma_q$.
        \end{itemize}
    \end{enumerate}
\end{lemma}

Recall that $r\parenv{\tilde{\bmx}}$ is the number of runs in $\tilde{\bmx}$. When $\abs{\bmy}\in\mathset{n+3,n+4}$, there are five cases for the relationship between $r\parenv{\tilde{\bmx}}$ and $r(\bmy)$.
\begin{itemize}
    \item If $\bmy$ is obtained from $\tilde{\bmx}$ by a $2$-burst-deletion in an alternating substring of length at least $3$, then $r(\tilde{\bmx})=r(\bmy)+2$.
    \item If $\bmy$ is obtained from $\tilde{\bmx}$ by deleting two consecutive symbols from an alternating substring of length at least $4$ and deleting one symbol from a run of length at least $3$, then $r(\tilde{\bmx})=r(\bmy)+2$.
    \item If $\bmy$ is obtained from $\tilde{\bmx}$ by replacing $\tilde{\bmx}_{[i-1,i+3]}$ with $\tilde{x}_{i-1}\tilde{x}_{i}$ for some $1\le i\le n$ with $\tilde{\bmx}_{[i-1,i+3]}\in\mathset{aabaa,abbab,abaab}$, where $a,b$ are distinct symbols from $\Sigma_q$, then $r(\tilde{\bmx})=r(\bmy)+2$.
    \item If $\bmy$ is obtained from $\tilde{\bmx}$ by replacing $\tilde{\bmx}_{[i-1,i+3]}$ with $\tilde{x}_{i-1}\tilde{x}_{i}$ for some $1\le i\le n$ satisfying $\tilde{\bmx}_{[i-1,i+3]}=abcab$ for some mutually distinct $a,b,c\in\Sigma_q$, then $r(\tilde{\bmx})=r(\bmy)+3$. This case is only possible when $q>2$.
    \item Otherwise, we have $r(\tilde{\bmx})=r(\bmy)$.
\end{itemize}

In each of first three cases, we have $r(\tilde{\bmx})=r(\bmy)+2$. The first case can be distinguished, however, since $\abs{\bmy}=n+4$ in this case. However, in both of the second and third cases, it holds that $r(\tilde{\bmx})=r(\bmy)+2$ and $\abs{\bmy}=n+3$. Hence, we need to introduce a new variable in $\tilde{\bmx}$ to help distinguish between these two cases.

Let $N_{\ge3}^{\textup{run}}(\tilde{\bmx})$ denote the total length of all runs in $\tilde{\bmx}$ whose lengths are at least $3$. It is straightforward to verify that, in the second case, $N_{\ge3}^{\textup{run}}(\tilde{\bmx})-N_{\ge3}^{\textup{run}}(\bmy)\in\{1,3\}$, whereas in the third case, $N_{\ge3}^{\textup{run}}(\tilde{\bmx})-N_{\ge3}^{\textup{run}}(\bmy)\in\{0,2\}$. Therefore, knowing $N_{\ge3}^{\textup{run}}\parenv{\tilde{\bmx}}\pmod{2}$ helps to distinguish between the second and third cases.

In the third and fourth cases, the deletions can be corrected by using \Cref{cor_TburstTperiodic} (setting $T=3$). In the second case, we will first correct the single deletion in the run and then correct the $2$-burst-deletion in the alternating substring.

For this purpose, we need some additional notations. Let $m_{\ge2}(\tilde{\bmx})$ denote the number of runs in $\tilde{\bmx}$ whose lengths are at least $2$. If $m_{\ge2}(\tilde{\bmx})=0$, set $h_{\textup{run},\ge2}\parenv{\tilde{\bmx},i}=0$ for all $-2\le i\le n+3$. If $m_{\ge2}(\tilde{\bmx})\ge1$, for all $-2\le i\le n+3$, set $h_{\textup{run},\ge2}\parenv{\tilde{\bmx},i}=j$ if $\tilde{x}_j$ belongs to the $j$-th run of length at least $2$, and $h_{\textup{run},\ge2}\parenv{\tilde{\bmx},i}=0$ if $\tilde{x}_j$ does not belong to any such runs. Define $H_{\textup{run},\ge2}\parenv{\tilde{\bmx}}=\sum_{i=-2}^{n+3}h_{\textup{run},\ge2}\parenv{\tilde{\bmx},i}$. This function helps to correct the single deletion in a run of length at least $3$.

We are ready to present constructions of codes. As in \Cref{subsec_l2t2,subsec_l3t2}, the binary case and non-binary case are treated in slightly different manners. We first give the construction of binary codes.

Let $F_{\textup{run},1}\parenv{\tilde{\bmx}}$, $F_{\textup{run},2}\parenv{\tilde{\bmx}}$ and $F_{\textup{run},3}\parenv{\tilde{\bmx}}$ be defined as in \Cref{subsec_l2t2}. Let $p_T(\cdot)$ be the function in \Cref{cor_TburstTperiodic}.
\begin{theorem}\label{thm_binaryl3t3}
    For $n\ge2$, let $p>2n+12$ be the smallest prime. For given $0\le a_1,a_2,a_3<p$, $0\le b<\floorenv{\frac{n+3}{2}}$, $0\le c<8\ceilenv{\frac{n+6}{3}}$, $0\le d_1<3$ and $d_2\in\{0,1\}$, set $\bma=\parenv{a_1,a_2}$ and $\bmd=(d_1,d_2)$. Define the code
     \begin{equation*}
        \cC(n;\bma,b,c,\bmd)=\mathset{\bmx\in\Sigma_2^n:
         \begin{array}{c}
        F_{\textup{run},k}\parenv{\tilde{\bmx}}\equiv a_k\pmod{p},\forall k\in[3]\\
        H_{\textup{run},\ge2}\parenv{\tilde{\bmx}}\equiv b\pmod{\floorenv{\frac{n+3}{2}}}\\
        \vt\parenv{\psi\parenv{p_3(\tilde{\bmx})}}\equiv c\pmod{8\ceilenv{\frac{n+6}{3}}}\\
        r(\tilde{\bmx})\equiv d_1\pmod{3}\\
        N_{\ge3}^{\textup{run}}\parenv{\tilde{\bmx}}\equiv d_2\pmod{2}
        \end{array}
        }.
    \end{equation*}
    This code can correct three deletions in $\cR_{3}(\tilde{\bmx})$. Moreover, there exist some $\bma$, $b$, $c$ and $\bmd$ such that $\rho\parenv{\cC(n;\bma,b,c,\bmd)}\le 3\log\parenv{p}+\log\parenv{\floorenv{\frac{n+3}{2}}}+\log\parenv{8\ceilenv{\frac{n+6}{3}}}+\log(3\cdot 2)=5\log n+O(1)$.
\end{theorem}
\begin{IEEEproof}
    Suppose that $\tilde{\bmx}=\bmu\bmx\bmv$ is transmitted. Let $\cR^\prime$ be the received vector of $3$-mers, where $\bmx\in\cC(n;\bma,b,c,\bmd)$ and $\cR^\prime$ is a subvector of $\cR_{3}(\tilde{\bmx})$ of length $n+1$. Let $\bmy$ be the sequence derived from $\cR^\prime$ as described above. We show how to recover $\tilde{\bmx}$ from $\bmy$.

    By \Cref{lem_l3t3}, if $|\bmy|=n+6$, then $\bmy=\tilde{\bmx}$. Therefore, it remains to consider the case $\abs{\bmy}\in\mathset{n+5,n+4,n+3}$.

    First suppose that $\abs{\bmy}=n+5$. By \Cref{lem_l3t3}, the sequence $\bmy$ is obtained from $\tilde{\bmx}$ by deleting one symbol from a run of length at least $2$. Hence $r(\tilde{\bmx})=r(\bmy)\le n+5$. As in the proof of \Cref{thm_binaryl2t2}, the sequence $\tilde{\bmx}$ can be recovered.

    Next, suppose that $\abs{\bmy}=n+4$. If $r(\bmy)\equiv d_1\pmod{3}$, then $\bmy$ is obtained from $\tilde{\bmx}$ by deleting one symbol from each of two runs of length at least $2$, or by deleting two symbols from a run of length at least $3$. In this case, $\tilde{\bmx}$ can be recovered as in the proof of \Cref{thm_binaryl2t2}, using the two  constraints $F_{\textup{run},1}\parenv{\tilde{\bmx}}\equiv a_1\pmod{p}$ and $F_{\textup{run},2}\parenv{\tilde{\bmx}}\equiv a_2\pmod{p}$.

    If $\abs{\bmy}=n+4$ and $r(\bmy)\equiv d_1-2\pmod{3}$, then $\bmy$ is obtained from $\tilde{\bmx}$ by a $2$-burst-deletion in an alternating substring of length at least $3$. We can recover $\tilde{\bmx}$ by using \Cref{lem_binary2burst} and the condition $F_{\textup{run},1}\parenv{\tilde{\bmx}}\equiv a_1\pmod{p}$. This also explains why we require $p>2n+12$ instead of $p\ge n+5$.

    Finally, suppose that $\abs{\bmy}=n+3$. If $r(\bmy)\equiv d_1\pmod{3}$, then $\bmy$ is obtained from $\tilde{\bmx}$ either by deleting one symbol from each of three runs of length at least $3$, or by deleting two symbols from a run of length at least $4$ and deleting one symbol from a run of length at least $3$, or by deleting three symbols from a run of length at least $5$. Therefore, $r\parenv{\tilde{\bmx}}=r(\bmy)\le n+3$. Suppose that the three deletions occurred in the $i$-th, $j$-th and $k$-th runs, where $i\le j\le k$. For $r\in[3]$, let $\delta_r$ be the unique integer in $[0,p-1]$ such that $\delta_r\equiv a_r-F_{\textup{run},r}\parenv{\tilde{\bmx}}\pmod{p}$. Then we have $i^r+j^r+k^r\equiv\delta_r\pmod{p}$ for $r\in[3]$. \Cref{lem_pre1} guarantees that we can uniquely and efficiently determine $i,j,k$ from these three congruences. The recovery of $\tilde{\bmx}$ is now straightforward.

    If $r(\bmy)\equiv d_1-2\pmod{3}$ and $N_{\ge3}^{\textup{run}}(\bmy)\not\equiv d_2\pmod{2}$, then $\bmy$ is obtained from $\tilde{\bmx}$ by a $2$-burst-deletion in an alternating substring of length at least $4$ and deleting one symbol from a run of length at least $3$. We first correct the single deletion in runs. Since the alternating substring has length at least $4$, the two deletions do affect runs of length at least $2$. As a result, we have $m_{\ge2}(\tilde{\bmx})=m_{\ge2}(\bmy)\le\floorenv{\abs{\bmy}/2}=\floorenv{\frac{n+3}{2}}$. Assume that the deletion occurred in the $i$-th run of length at least $2$. Then $i$ is the unique integer in $\sparenv{\floorenv{\frac{n+3}{2}}}$ satisfying $i\equiv b-H_{\textup{run},\ge2}(\bmy)\pmod{\floorenv{\frac{n+3}{2}}}$. Once knowing $i$, we extend the $i$-th run of length at least $2$ in $\bmy$ by one symbol. The resulting sequence is denoted by $\bmz$, which is obtained from $\tilde{\bmx}$ by a $2$-burst-deletion in an alternating substring of length at least $4$. Now $\tilde{\bmx}$ can be recovered by \Cref{lem_binary2burst}.

    If $r(\bmy)\equiv d_1-2\pmod{3}$ and $N_{\ge3}^{\textup{run}}(\bmy)\equiv d_2\pmod{2}$, then $\bmy$ is obtained from $\tilde{\bmx}$ by a $3$-burst-deletion in a $3$-periodic substring of length at least $5$. This case can be handled with the help of \Cref{cor_TburstTperiodic} and the condition $\vt\parenv{\psi\parenv{p_3(\tilde{\bmx})}}\equiv c\pmod{8\ceilenv{\frac{n+6}{3}}}$. Now the proof is completed.
\end{IEEEproof}

Next, we give the construction of non-binary codes.
\begin{theorem}\label{thm_qaryl3t3}
    Let $q>2$ and $n\ge2$. Let $p\ge n+5$ be the smallest prime. For given $0\le a_1,a_2,a_3<p$, $0\le b<\floorenv{\frac{n+3}{2}}$, $0\le c_2<q^2\ceilenv{\frac{n+6}{2}}$, $0\le c_3<q^3\ceilenv{\frac{n+6}{3}}$, $0\le d_1<4$ and $d_2\in\{0,1\}$, set $\bma=\parenv{a_1,a_2}$, $\bmc=\parenv{c_1,c_2}$ and $\bmd=(d_1,d_2)$. Define the code
     \begin{equation*}
        \cC(n;\bma,b,\bmc,\bmd)=\mathset{\bmx\in\Sigma_q^n:
         \begin{array}{c}
        F_{\textup{run},k}\parenv{\tilde{\bmx}}\equiv a_k\pmod{p},\forall k\in[3]\\
        H_{\textup{run},\ge2}\parenv{\tilde{\bmx}}\equiv b\pmod{\floorenv{\frac{n+3}{2}}}\\
        \vt\parenv{\psi\parenv{\tilde{\bmx}}}\equiv c_1\pmod{q(n+6)}\\
        \vt\parenv{\psi\parenv{p_2(\tilde{\bmx})}}\equiv c_2\pmod{q^2\ceilenv{\frac{n+6}{2}}}\\
        \vt\parenv{\psi\parenv{p_3(\tilde{\bmx})}}\equiv c_3\pmod{q^3\ceilenv{\frac{n+6}{3}}}\\
        r(\tilde{\bmx})\equiv d_1\pmod{4}\\
        N_{\ge3}^{\textup{run}}\parenv{\tilde{\bmx}}\equiv d_2\pmod{2}
        \end{array}
        }.
    \end{equation*}
    This code can correct three deletions in $\cR_{3}(\tilde{\bmx})$. Moreover, there exist some $\bma$, $b$, $\bmc$ and $\bmd$ such that $\rho\parenv{\cC(n;\bma,b,\bmc,\bmd)}\le 3\log\parenv{p}+\log\parenv{\floorenv{\frac{n+3}{2}}}+\log\parenv{q(n+6)}+\log\parenv{q^2\ceilenv{\frac{n+6}{2}}}+\log\parenv{q^3\ceilenv{\frac{n+6}{3}}}+\log(4\cdot 2)=7\log n+O(1)$.
\end{theorem}
\begin{IEEEproof}
    Suppose that $\tilde{\bmx}=\bmu\bmx\bmv$ is transmitted, and let $\cR^\prime$ be the received vector of $3$-mers, where $\bmx\in\cC(n;\bma,b,\bmc,\bmd)$ and $\cR^\prime$ is a subvector of $\cR_{3}(\tilde{\bmx})$ of length $n+1$. Let $\bmy$ be the sequence derived from $\cR^\prime$ as described above. We show how to recover $\tilde{\bmx}$ from $\bmy$. If $|\bmy|=n+6$, then $\bmy=\tilde{\bmx}$.

    If $\abs{\bmy}=n+5$, by \Cref{lem_l3t3}, the sequence $\bmy$ is obtained from $\tilde{\bmx}$ by one deletion. Hence  $\tilde{\bmx}$ can be recovered by using \Cref{lem_qary1del} and the condition $\vt\parenv{\psi\parenv{\tilde{\bmx}}}\equiv c_1\pmod{q(n+6)}$.

    Next, suppose that $\abs{\bmy}=n+4$. If $r(\bmy)\equiv d_1\pmod{4}$, then $\bmy$ is obtained from $\tilde{\bmx}$ by deleting one symbol from each of two runs of length at least $2$, or by deleting two symbols from a run of length at least $3$. In this case, $\tilde{\bmx}$ can be recovered as in the proof of \Cref{thm_binaryl2t2}, using the two  constraints $F_{\textup{run},1}\parenv{\tilde{\bmx}}\equiv a_1\pmod{p}$ and $F_{\textup{run},2}\parenv{\tilde{\bmx}}\equiv a_2\pmod{p}$.

    If $\abs{\bmy}=n+4$ and $r(\bmy)\equiv d_1-2\pmod{4}$, then $\bmy$ is obtained from $\tilde{\bmx}$ by a $2$-burst-deletion in an alternating substring of length at least $3$. We can recover $\tilde{\bmx}$ by using \Cref{cor_TburstTperiodic} and the condition $\vt\parenv{\psi\parenv{p_2(\tilde{\bmx})}}\equiv c_2\pmod{q^2\ceilenv{\frac{n+6}{2}}}$.

    Finally, suppose that $\abs{\bmy}=n+3$. If $r(\bmy)\equiv d_1\pmod{3}$, then $\bmy$ is obtained from $\tilde{\bmx}$ either by deleting one symbol from each of three runs of length at least $3$, or by deleting two symbols from a run of length at least $4$ and deleting one symbol from a run of length at least $3$, or by deleting three symbols from a run of length at least $5$. Therefore, $r\parenv{\tilde{\bmx}}=r(\bmy)\le n+3$. Suppose that the three deletions occurred in the $i$-th, $j$-th and $k$-th runs, respectively, where $i\le j\le k$. For $r\in[3]$, let $\delta_r$ be the unique integer in $[0,p-1]$ such that $\delta_r\equiv a_r-F_{\textup{run},r}\parenv{\tilde{\bmx}}\pmod{p}$. Then we have $i^r+j^r+k^r\equiv\delta_r\pmod{p}$ for $r\in[3]$. \Cref{lem_pre1} guarantees that we can uniquely and efficiently determine $i,j,k$ from these three congruences. Now the recovery of $\tilde{\bmx}$ is straightforward.

    If $r(\bmy)\equiv d_1-2\pmod{4}$ and $N_{\ge3}^{\textup{run}}(\bmy)\not\equiv d_2\pmod{2}$, then $\bmy$ is obtained from $\tilde{\bmx}$ by a $2$-burst-deletion in an alternating substring of length at least $4$ and deleting one symbol from a run of length at least $3$. We first correct the single deletion in runs. Since the alternating substring has length at least $4$, the two deletions do affect runs of length at least $2$. As a result, we have $m_{\ge2}(\tilde{\bmx})=m_{\ge2}(\bmy)\le\floorenv{\abs{\bmy}/2}=\floorenv{\frac{n+3}{2}}$. Assume that the deletion occurred in the $i$-th run of length at least $2$ (In fact, this run has length at least $3$). Then $i$ is the unique integer in $\sparenv{\floorenv{\frac{n+3}{2}}}$ satisfying $i\equiv b-H_{\textup{run},\ge2}(\bmy)\pmod{\floorenv{\frac{n+3}{2}}}$. Once knowing $i$, we extend the $i$-th run of length at least $2$ in $\bmy$ by one symbol. The resulting sequence is denoted by $\bmz$, which is obtained from $\tilde{\bmx}$ by a $2$-burst-deletion in an alternating substring of length at least $4$. Now $\tilde{\bmx}$ can be recovered by using \Cref{cor_TburstTperiodic} and the condition $\vt\parenv{\psi\parenv{p_2(\tilde{\bmx})}}\equiv c_2\pmod{q^2\ceilenv{\frac{n+6}{2}}}$.

    If $r(\bmy)\equiv d_1-2\pmod{4}$ and $N_{\ge3}^{\textup{run}}(\bmy)\equiv d_2\pmod{2}$, or $r(\bmy)\equiv d_1-3\pmod{4}$, then $\bmy$ is obtained from $\tilde{\bmx}$ by a $3$-burst-deletion in a $3$-periodic substring of length at least $5$. This case can be handled with the help of \Cref{cor_TburstTperiodic} and the condition $\vt\parenv{\psi\parenv{p_3(\tilde{\bmx})}}\equiv c\pmod{8\ceilenv{\frac{n+6}{3}}}$. This completes the proof.
\end{IEEEproof}

\section{Explicit Codes}\label{sec_explicitcode}
The codes in \Cref{sec_genconstruction,sec_sporadic} are non-explicit in the sense that no encoding algorithms are provided. In this section, we give explicit codes based on the constructions in the preceding two sections.

Let $q,N\ge2$ be two integers. For any $m\in[0,N-1]$, we can represent $m$ as a unique $q$-ary sequence of length $\ceilenv{\log_q(N)}$. Denote this $q$-ary sequence by $m_{q\text{-ary}}$
\begin{proposition}\label{prop_explicit}
    Let $q\ge2$, $\ell$ and $t$ be integers. Let $\bmu$ and $\bmv$ be two given $\ell$-mers. Suppose that there is an explicit function $f:\Sigma_q^{m}\rightarrow[0,N(m)-1]$, where $N(m)$ is an integer depending on $m$, with the following property:
    \begin{itemize}
        \item for any $\bmz\in\Sigma_q^m$, when given $f(\bmz)$, one can efficiently recover $\bmz$ from any $\cR^\prime$ that is obtained from $\cR_{\ell}(\bmz)$ by deleting $t$ $\ell$-mers, as long as the first and last $\ell$-mer are not deleted.
    \end{itemize}
    For any $\bmx\in\Sigma_q^n$, let $E(\bmx)=\parenv{\bmx,0^{\ell},f\parenv{\bmu\bmx 0^{\ell}}_{q\text{-ary}},0^{\ell},\Rep_{t+1}\parenv{f\parenv{0^{\ell}f\parenv{\bmu\bmx 0^{\ell}}_{q\text{-ary}}0^{\ell}}_{q\text{-ary}}}}$, where $\Rep_{t+1}(\bmz)$ is obtained from $\bmz$ by repeating each symbol $t+1$ times.
     Define the code
    $$
    \cC_f=\mathset{E(\bmx):\bmx\in\Sigma_q^n}.
    $$
    Then $\cC_f$ is an $\ell$-read $t$-deletion correcting code, with redundancy $\log(N(m))+O\parenv{\log\log(N(m))}$.
\end{proposition}
\begin{IEEEproof}
    Suppose that $\bmu\bmc\bmv$ is transmitted, where $\bmc=E(\bmx)$ for some $\bmx\in\Sigma_q^n$. Let $\cR^\prime$ be obtained from $\cR_{\ell}(\bmu\bmc\bmv)$ by deleting $t$ $\ell$-mers. We show how to recover $\bmx$ from $\cR^\prime$.

    Denote $M=\abs{\bmc}$, $n_1=\abs{f\parenv{\bmu\bmx 0^{\ell}}_{q\text{-ary}}}$ and $n_2=\abs{\Rep_{t+1}\parenv{f\parenv{0^{\ell}f\parenv{\bmu\bmx 0^{\ell}}_{q\text{-ary}}0^{\ell}}_{q\text{-ary}}}}$. Then $M=n+n_1+n_2+2\ell$.
    Denote $\cR^\prime=\parenv{\bms_0,\bms_1,\ldots,\bms_{M+\ell-t}}$. Let
    \begin{equation*}
        \begin{array}{c}
            \cR_1=\parenv{\bms_0,\bms_1,\ldots,\bms_{n+\ell-t}},\\
            \cR_2=\parenv{\bms_{n+\ell},\ldots,\bms_{n+n_1+2\ell-t}},\\
            \cR_3=\parenv{\bms_{n+n_1+2\ell},\ldots,\bms_{M+\ell-t}}.
        \end{array}
    \end{equation*}
    It is straightforward to verify that the three vectors $\cR_1$, $\cR_2$ and $\cR_3$ are obtained from $\cR_{\ell}\parenv{\bmu\bmx0^{\ell}}$, $\cR_{\ell}\parenv{0^{\ell}f\parenv{\bmu\bmx 0^{\ell}}_{q\text{-ary}}0^{\ell}}$ and $\cR_{\ell}\parenv{0^{\ell}\Rep_{t+1}\parenv{f\parenv{0^{\ell}f\parenv{\bmu\bmx 0^{\ell}}_{q\text{-ary}}0^{\ell}}_{q\text{-ary}}}\bmv}$ by deleting $t$ $\ell$-mers, respectively. Moreover, since the first and last $\ell$-mers of each of $\bmu\bmx0^{\ell}$, $0^{\ell}f\parenv{\bmu\bmx 0^{\ell}}_{q\text{-ary}}0^{\ell}$ and $0^{\ell}\Rep_{t+1}\parenv{f\parenv{0^{\ell}f\parenv{\bmu\bmx 0^{\ell}}_{q\text{-ary}}0^{\ell}}_{q\text{-ary}}}\bmv$ are known, we can assume without loss of generality that the first and last $\ell$-mers of each of the three substrings are not deleted.

    By \Cref{thm_multiburst}, we can obtain a consistent vector $\cR_3^\prime$ of $\ell$-mers from $\cR_3$. Let $\bmy=\cR_{\ell}^{-1}\parenv{\cR_3^\prime}$. Then $\bmy$ is obtained from $0^{\ell}\Rep_{t+1}\parenv{f\parenv{0^{\ell}f\parenv{\bmu\bmx 0^{\ell}}_{q\text{-ary}}0^{\ell}}_{q\text{-ary}}}\bmv$ by deleting at most $t$ symbols. In addition, the prefix $0^{\ell}$ and the suffix $\bmv$ are not affected by the deletions. Therefore, we can recover $\Rep_{t+1}\parenv{f\parenv{0^{\ell}f\parenv{\bmu\bmx 0^{\ell}}_{q\text{-ary}}}}$ and hence $f\parenv{0^{\ell}f\parenv{\bmu\bmx 0^{\ell}}0^{\ell}}$ immediately. Next, by the property of $f$, one can efficiently recover $f\parenv{\bmu\bmx 0^{\ell}}$ from $\cR_2$. Finally, the original $\bmx$ can be recovered by the same argument.
\end{IEEEproof}

By \Cref{prop_explicit}, each construction in \Cref{thm_1del,thm_multidel1,thm_multidel2,thm_binaryl2t2,thm_qaryl2t2,thm_qaryl3t2,thm_binaryl3t3,thm_qaryl3t3} and \Cref{cor_bianryl3t2} gives an explicit code, since each construction provides a function with the property required in \Cref{prop_explicit}. Here, we take \Cref{thm_multidel1} as an example to illustrate the idea. Let $F_{1},\ldots,F_t$ be the $t$ functions defined prior to \Cref{thm_multidel1}. Let $p\ge\max\mathset{m-\ell+t,t+1}$ be the smallest prime. For $\bmx\in\Sigma_q^m$, define $f(\bmx)=\parenv{F_1(\bmx)\pmod{p},\ldots,F_t(\bmx)\pmod{p}}\in\sparenv{0,p^t-1}$. The proof of \Cref{thm_multidel1} shows that the function $f$ has the property required in \Cref{prop_explicit}.

\section{Conclusion}\label{sec_conclusion}
In this paper, we studied deletion-correcting codes for the $\ell$-symbol read channel under an adversarial deletion model. The main technical contribution is a structural description of how deletions in the read vector affect the underlying transmitted sequence. After restoring consistency by inserting a minimum number of $\ell$-mers, the resulting sequence can be viewed as the transmitted sequence after deletions occurring inside certain periodic substrings. This observation led to the notion of check patterns and provides a unified way to track the hidden effect of deleted reads.

Using this framework, we constructed several families of $\ell$-read $t$-deletion correcting codes. The general constructions show that redundancy $t\log n+O(1)$ is achievable when $\ell/2\ge t\ge1$, while the boundary case $\ell=2t-1$ can still be handled with redundancy $(2t-1)\log n+O(1)$ for $t\ge3$. These constructions rely on the fact that when $\ell/2\ge t$ or $\ell=2t-1$, deletions in the $\ell$-read vector only shorten certain check patterns in the original sequence. When $(\ell,t)$ does not satisfy these conditions, this phenomenon no longer holds. However, the sporadic constructions in \Cref{sec_sporadic} show that these conditions are not necessary for achieving these redundancies.

Several questions remain open. One primary direction is to determine the optimal redundancy for general pairs $(\ell,t)$, especially when $t>\ell/2$, where the check-pattern stability arguments used in the general construction no longer apply directly. It would also be useful to develop systematic and computationally efficient encoding algorithms for the proposed codes in \Cref{sec_genconstruction,sec_sporadic}. Another natural extension is to study joint correction of deletions, insertions, and substitutions in the $\ell$-symbol read channel.

\section*{Acknowledgment}
The first author thanks Yubo Sun for bringing references \cite{HananiaYaakobi2025ISIT,SimaBruck2023IT} to him.

\bibliographystyle{IEEEtran}
\bibliography{ref.bib}

\begin{thebibliography}{10}
\providecommand{\url}[1]{#1}
\csname url@samestyle\endcsname
\providecommand{\newblock}{\relax}
\providecommand{\bibinfo}[2]{#2}
\providecommand{\BIBentrySTDinterwordspacing}{\spaceskip=0pt\relax}
\providecommand{\BIBentryALTinterwordstretchfactor}{4}
\providecommand{\BIBentryALTinterwordspacing}{\spaceskip=\fontdimen2\font plus
\BIBentryALTinterwordstretchfactor\fontdimen3\font minus \fontdimen4\font\relax}
\providecommand{\BIBforeignlanguage}[2]{{%
\expandafter\ifx\csname l@#1\endcsname\relax
\typeout{** WARNING: IEEEtran.bst: No hyphenation pattern has been}%
\typeout{** loaded for the language `#1'. Using the pattern for}%
\typeout{** the default language instead.}%
\else
\language=\csname l@#1\endcsname
\fi
#2}}
\providecommand{\BIBdecl}{\relax}
\BIBdecl

\bibitem{Omer2024TMBMSC}
O.~Sabary, H.~M. Kiah, P.~H. Siegel, and E.~Yaakobi, ``Survey for a {D}ecade of {C}oding for {DNA} {S}torage,'' \emph{IEEE Trans. Mol. Biol. Multi-Scale Commun.}, vol.~10, no.~2, pp. 253--271, Jun. 2024.

\bibitem{Milenkovic2024TCOMM}
O.~Milenkovic and C.~Pan, ``{DNA}-based data storage systems: A review of implementations and code constructions,'' \emph{IEEE Trans. Commun.}, vol.~72, no.~7, pp. 3803--3828, Jul. 2024.

\bibitem{Deamer2016}
D.~Deamer, M.~Akeson, and D.~Branton, ``Three decades of nanopore sequencing,'' \emph{Nature Biotechnol.}, vol.~34, no.~5, pp. 518--524, May 2016.

\bibitem{Hulett2021ISIT}
R.~Hulett, S.~Chandak, and M.~Wootters, ``On coding for an abstracted nanopore channel for {DNA} storage,'' in \emph{Proc. IEEE Int. Symp. Inf. Theory (ISIT)}, Melbourne, Australia, Jul. 2021, pp. 2465--2470.

\bibitem{Maowei2018IT}
W.~Mao, S.~N. Diggavi, and S.~Kannan, ``Models and information-theoretic bounds for nanopore sequencing,'' \emph{IEEE Trans. Inf. Theory}, vol.~64, no.~4, pp. 3216--3236, Apr. 2018.

\bibitem{Banerjee2023ISIT}
A.~Banerjee, Y.~Yehezkeally, A.~Wachter-Zeh, and E.~Yaakobi, ``Error-correcting codes for nanopore sequencing,'' in \emph{Proc. IEEE Int. Symp. Inf. Theory (ISIT)}, Taipei, Taiwan, Jun. 2023, pp. 364--369.

\bibitem{Banerjee2024IT}
------, ``Error-correcting codes for nanopore sequencing,'' \emph{IEEE Trans. Inf. Theory}, vol.~70, no.~7, pp. 4956--4967, Jul. 2024.

\bibitem{Yubo2025ITnanopore}
Y.~Sun and G.~Ge, ``Bounds and constructions of $\ell$-read codes under the hamming metric,'' \emph{IEEE Trans. Inf. Theory}, vol.~71, no.~8, pp. 5868--5883, Aug. 2025.

\bibitem{Banerjee2025ISIT}
A.~Banerjee, Y.~Yehezkeally, A.~Wachter-Zeh, and E.~Yaakobi, ``Correcting multiple substitutions in nanopore-sequencing reads,'' in \emph{Proc. IEEE Int. Symp. Inf. Theory (ISIT)}, Ann Arbor, MI, USA, Jun. 2025.

\bibitem{Banerjee2024ISIT}
------, ``Correcting a single deletion in reads from a nanopore sequencer,'' in \emph{Proc. IEEE Int. Symp. Inf. Theory (ISIT)}, Athens, Greece, Jul. 2024, pp. 103--108.

\bibitem{WenjunYu2025IT}
W.~Yu, Z.~Ye, and M.~Schwartz, ``On the asymptotic rate of optimal codes that correct tandem duplications for nanopore sequencing,'' \emph{IEEE Trans. Inf. Theory}, vol.~71, no.~5, pp. 3569--3581, May 2025.

\bibitem{CheeYeowMeng2024ISIT}
Y.~M. Chee, K.~A.~S. Immink, and V.~K. Vu, ``Coding scheme for noisy nanopore sequencing with backtracking and skipping errors,'' in \emph{2024 IEEE International Symposium on Information Theory (ISIT)}, Athens, Greece, Jul. 2024, pp. 458--463.

\bibitem{Zitan2025ITW}
H.~Xie and Z.~Chen, ``Two-deletion correcting codes for nanopore sequencing,'' in \emph{Proc. IEEE Inf. Theory Workshop (ITW)}, Sydney, Australia, Sept. 2025, pp. 1--6.

\bibitem{Zitan202601}
------, ``Deletion-correcting codes for an adversarial nanopore channel,'' \emph{arXiv:2601.21236v2}, 2026.

\bibitem{Zitan202606}
------, ``First-order optimal codes for an adversarial nanopore channel,'' \emph{arXiv:2601.21236v3}, 2026.

\bibitem{CassutoBlaum2011IT}
Y.~Cassuto and M.~Blaum, ``Codes for symbol-pair read channels,'' \emph{IEEE Trans. Inf. Theory}, vol.~57, no.~12, pp. 8011--8020, Dec. 2011.

\bibitem{YaakobiBruckSiegel2016IT}
E.~Yaakobi, J.~Bruck, and P.~H. Siegel, ``Constructions and decoding of cyclic codes over $b$ -symbol read channels,'' \emph{IEEE Trans. Inf. Theory}, vol.~62, no.~4, pp. 1541--1551, Apr. 2016.

\bibitem{YeowVan2020ISIT}
Y.~M. Chee and V.~K. Vu, ``Codes correcting synchronization errors for symbol-pair read channels,'' in \emph{Proc. IEEE Int. Symp. Inf. Theory (ISIT)}, no. 746--750, Los Angeles, CA, USA, Jun. 2020.

\bibitem{CheeYeowMeng2021IT}
Y.~M. Chee, T.~Etzion, H.~M. Kiah, S.~Marcovich, A.~Vardy, V.~Khu~Vu, and E.~Yaakobi, ``Locally-constrained de bruijn codes: Properties, enumeration, code constructions, and applications,'' \emph{IEEE Trans. Inf. Theory}, vol.~67, no.~12, pp. 7857--7875, Dec. 2021.

\bibitem{Hanania2025IT}
D.~Hanania, D.~Bar-Lev, Y.~Nogin, Y.~Shechtman, and E.~Yaakobi, ``On the capacity of {DNA} labeling,'' \emph{IEEE Trans. Inf. Theory}, vol.~71, no.~5, pp. 3457--3472, May 2025.

\bibitem{HofmeisterChristoph2026IT}
C.~Hofmeister, A.~Gruica, D.~Hanania, R.~Bitar, and E.~Yaakobi, ``Achieving {DNA} labeling capacity with minimum labels through extremal de bruijn subgraphs,'' \emph{IEEE Trans. Inf. Theory}, vol.~72, no.~2, pp. 1122--1132, Feb. 2026.

\bibitem{HananiaYaakobi2025ISIT}
D.~Hanania and E.~Yaakobi, ``Error-correcting codes for labeled {DNA} sequences,'' in \emph{Proc. IEEE Int. Symp. Inf. Theory (ISIT)}, Ann Arbor, MI, USA, Jun. 2025.

\bibitem{ZhangChao2015ISCA}
C.~Zhang, G.~Sun, X.~Zhang, W.~Zhang, W.~Zhao, T.~Wang, Y.~Liang, Y.~Liu, Y.~Wang, and J.~Shu, ``Hi-fi playback: Tolerating position errors in shift operations of racetrack memory,'' in \emph{Proc. 42nd ACM/IEEE Int. Symp. Comput. Archit. (ISCA)}, Portland, OR, USA, Jun. 2015, pp. 694--706.

\bibitem{CheeKiahVardy2018IT}
Y.~M. Chee, H.~M. Kiah, A.~Vardy, V.~K. Vu, and E.~Yaakobi, ``Coding for racetrack memories,'' \emph{IEEE Trans. Inf. Theory}, vol.~64, no.~11, pp. 7094--7112, Nov. 2018.

\bibitem{CheeYeowMeng2018ITW}
Y.~M. Chee, R.~Gabrys, A.~Vardy, V.~K. Vu, and E.~Yaakobi, ``Reconstruction from deletions in racetrack memories,'' in \emph{Proc. IEEE Inf. Theory Workshop (ITW)}, Guangzhou, China, Nov. 2018, pp. 1--5.

\bibitem{SimaBruck2023IT}
J.~Sima and J.~Bruck, ``Correcting multiple deletions and insertions in racetrack memory,'' \emph{IEEE Trans. Inf. Theory}, vol.~69, no.~9, p. 56195639, Sept. 2023.

\bibitem{FineWilf1965}
N.~J. Fine and H.~S. Wilf, ``Uniqueness theorems for periodic functions,'' \emph{Proc. Amer. Math. Soc.}, vol.~16, no.~1, pp. 109--114, Feb. 1965.

\bibitem{kolpakov1999FOCS}
R.~Kolpakov and G.~Kucherov, ``Finding maximal repetitions in a word in linear time,'' in \emph{Proc. 40th Annu. Symp. Found. Comput. Sci. (FOCS)}, New York, NY, USA, Oct. 1999, pp. 596--604.

\bibitem{Lara2010SIAM}
L.~Dolecek and V.~Anantharam, ``Repetition error correcting sets: Explicit constructions and prefixing methods,'' \emph{SIAM J. Discrete Math.}, vol.~23, no.~4, pp. 2120--–2146, 2010.

\bibitem{VL1967}
V.~I. Levenshtein, ``Asymptotically optimum binary code with correction for losses of one or two adjacent bits,'' \emph{Problemy Kibernetiki}, vol.~19, pp. 293--298, 1967.

\bibitem{Tuan2024IT}
T.~Thanh~Nguyen, K.~Cai, and P.~H. Siegel, ``A new version of q-ary {V}arshamov-{T}enengolts codes with more efficient encoders: the differential {VT} codes and the differential shifted {VT} codes,'' \emph{IEEE Trans. Inf. Theory}, vol.~70, no.~10, pp. 6989--7004, Oct. 2024.

\bibitem{Shuche2024IT}
S.~Wang, Y.~Tang, J.~Sima, R.~Gabrys, and F.~Farnoud, ``Non-binary codes for correcting a burst of at most $t$ deletions,'' \emph{IEEE Trans. Inf. Theory}, vol.~70, no.~2, pp. 964--979, Feb. 2024.

\bibitem{Yubo2025IT}
Y.~Sun, Z.~Lu, Y.~Zhang, and G.~Ge, ``Asymptotically optimal codes for $(t, s)$-burst error,'' \emph{IEEE Trans. Inf. Theory}, vol.~71, no.~3, pp. 1570--1584, Mar. 2025.

\end{thebibliography}
\end{document}